\documentclass[a4paper,cleveref, autoref, thm-restate]{lipics-v2021}

\newif\ifanonymous

\newif\ifnotkurz
\notkurztrue

\hideLIPIcs  
\makeatletter
\let\@oddfoot\@empty 
\makeatother

\nolinenumbers

\usepackage{amsmath}
\usepackage{arydshln}
\usepackage{amsfonts}
\usepackage{verbatim}

\usepackage{amsthm}

\usepackage[utf8]{inputenc}
\usepackage{graphicx}
\usepackage[T1]{fontenc}
\usepackage{tikz}
\usetikzlibrary{positioning,arrows.meta}
\usetikzlibrary{fit}
\usepackage{hyperref}
\usepackage{xurl} 

\newcommand\eps{\varepsilon}
\newcommand\eeta{f} 

\newcommand\stepref[1]{\hyperref[sec:step#1]{Step~#1}}

\newcommand\refToStepA[1]{\hyperref[para:step#1]{Step~#1}}
\newcommand\refToStep[1]{\hyperref[para:step#1]{#1}}

\ifanonymous
 \newcommand\clausEATCS{claus80-Bull-EATCS-anonymous}
\else
 \newcommand\clausEATCS{claus80-Bull-EATCS}
\fi

        \title{Probabilistic Finite Automaton Emptiness is
          Undecidable
          for a Fixed Automaton}

\author{G\"unter Rote}{Freie Universit\"at Berlin, Institut für
  Informatik,
Takustr.~
 9, 14195 Berlin, Germany
}{rote@inf.fu-berlin.de}{https://orcid.org/0000-0002-0351-5945}{}

\authorrunning{G. Rote}

\Copyright{G\"unter Rote}
        
\titlerunning{G. Rote: PFA
  Emptiness is
          Undecidable
        for a Fixed Automaton}
\authorrunning{G. Rote: PFA
  Emptiness is
          Undecidable
        for a Fixed Automaton}

      

\EventShortTitle{STACS
  ,submitted}
\EventAcronym{STACS}
\EventYear{2025}
\ArticleNo{32}
\EventLongTitle{42nd International Symposium on Theoretical Aspects of Computer Science}
\EventLocation{Jena}

\begin{document}

  \def\sectionautorefname{Section} 
  \def\subsectionautorefname{Section} 

\maketitle

\let\paragraph=\subparagraph

\begin{abstract}
  We construct a 
  probabilistic finite automaton (PFA) with 7 states
  and an input alphabet of 5~symbols
for which the PFA Emptiness Problem is undecidable.
The only 
input for the decision problem is the starting distribution.
For the proof, we use reductions from special instances of the Post
Correspondence Problem.

We also consider some variations:
The input alphabet of the PFA can be restricted to a binary alphabet
at the expense of a larger
number of states.  If we allow a rational output value 
for each state 
instead of a yes-no acceptance decision,
the number
of states can even be reduced to 6. 

\end{abstract}


{ \smallskip
  \vfill 
  \noindent
  \rlap{\color{lipicsLineGray}\vrule width\textwidth height 1pt }%
  \hspace*{7mm}\fboxsep1.5mm\colorbox[rgb]{1,1,1}{\raisebox{-0.4ex}{%
    \large\selectfont\sffamily\bfseries Contents}}%
\def\section*#1{}

 \def\addvspace#1{\vspace {0.03pt}}
 
                          

  

 \baselineskip = 12.7 pt
 \enlargethispage{1.2\baselineskip} 
  \tableofcontents
  \par
}



\newpage
\section{Probabilistic finite automata (PFA)}
\label{sec:PFA}

A probabilistic finite automaton (PFA)
combines characteristics of
a finite automaton and a Markov chain.
We give a formal definition
below. 
Informally, we can think of a PFA in terms of an algorithm that
reads a sequence of input symbols 
from left to right,
having only finite memory. That is, it can manipulate a
finite number of variables with bounded range, just like an ordinary
finite automaton. In addition, a
PFA 
can make coin flips.
As a consequence,
the question whether the PFA arrives in an accepting state and thus accepts a given input word 
is not a yes/no decision, but it happens with a certain {probability}.
The language \emph{recognized}
(or \emph{represented})
by a PFA is defined by specifying a
probability threshold or
\emph{cutpoint} $\lambda$. By convention, the language consists of all words for which the
probability of acceptance strictly exceeds~$\lambda$.

The \emph{PFA Emptiness Problem} is the problem of deciding whether
this language is empty.
This problem is undecidable.
Many other undecidability results rely on the undecidability of
the PFA Emptiness Problem, for example, problems about
matrix products~\cite{blondel2000},
growth problems
~\cite[Chapter~6]{bui23-phd},
or probabilistic planning problems~\cite
{jair03}. 

We present some 
sharpenings of the undecidability statement,
where certain parameters of the PFA are restricted.

 \subsection{Formal problem definition}
\label{sec:formal}

Formally, a PFA is given by a sequence of stochastic \emph{transition
 matrices} $M_\sigma$, 
one for each letter $\sigma$ from the input alphabet $\Sigma$.
The states correspond to the rows and columns of the matrices.
Thus,
 the matrices are $d\times d$ matrices for a PFA with $d$~states.
 The start state is chosen according
 to a given
 probability distribution $\pi\in \mathbb R^d$.
 The set of accepting
 states is characterized by a 0-1-vector $\eeta \in \{0,1\}^{d}$.
%
%
In terms of these data, the PFA Emptiness Problem
with cutpoint $\lambda$ is
as follows:

\goodbreak

\begin{quote}
  \textsc{PFA Emptiness.}
  Given  a 
  set of $k$ 
  stochastic matrices
 $\mathcal{M}=\{M_1,\ldots,M_k\} \subset  \mathbb Q^{d\times d}$,
a probability distribution $\pi\in \mathbb Q^d$,
and a 0-1-vector $\eeta \in \{0,1\}^{d}$, 
is there a sequence
$i_1,i_2,\ldots, i_m$
with $1\le i_j\le k$ for $j=1,\ldots,m$
such that
\begin{equation}\label{eq:accept}
  \pi^T\!M_{i_1}M_{i_2}\ldots M_{i_m}\eeta > \lambda \ ?
\end{equation}
\goodbreak
\end{quote}
The most natural choice is $\lambda=\frac12$, but
the problem is undecidable for
any fixed (rational or irrational) cutpoint $\lambda$ with $0<
\lambda< 1$.  We can also ask $\ge\lambda$ instead of $>\lambda$.

\paragraph{Rational-Weigthed Automata.}
A \emph{weighted automaton} over the reals or
\emph{generalized probabilistic automaton} or
\emph{pseudo-stochastic automaton}
is similar, but
the start vector $\pi$, the end vector $f$, and
the matrices~$M_i$ can be filled with arbitrary real numbers.

\begin{table}[thb]
  \centering
  \newcommand\inp{\emph{input}}
  \newcommand\fix{\emph{fixed}}
  \begin{tabular}{|l|crlcl|}
    \hline
    Theorem&$\pi$&\!\!\!\!\!\!
                   $|\mathcal{M}|$&$M \in
                                    \mathcal{M}$&$\eeta$&acceptance
                                                          crit.
    \\
  \hline
Thm.~\ref{thm-matrices-forward}\ref{4-positive}
           & \inp & 4 & $7\times7$, positive, \inp 
                                               &$f=e_1$ & $>1/7$ or $\ge 1/7$\\
Thm.~\ref{thm-matrices-forward}\ref{5-fixed}
           & \inp & 6 & $7\times7$, positive, \fix 
                                               &$f=e_1$ & $>1/7$ or $\ge 1/7$\\
Thm.~\ref{thm-matrices-forward}\ref{2-matrices}
           & \inp & 2 & $18\times18$, positive, \inp 
                                               &$f=e_1$
                                         & $>1/18$ or $\ge 1/18$\\[0pt]
    Thm.~\ref{thm-matrices-forward}\ref{2-matrices-fixed}
           & \inp & 2 & $28\times28$, positive, \fix 
           &$f=e_1$ 
           & $>1/28$ or $\ge 1/28$\\
\hline    
    Thm.~\ref{thm-matrices-backward}\ref{backward}
           & $\pi= e_2$
 & 4 & $6\times 6$, positive, \inp
     &$\hskip -1.65em
     f\in [0,1]^6$, \inp
                           & $>\lambda$ or $\ge \lambda$
\\
    Thm.~\ref{thm-matrices-backward}\ref{backward-fixed}
      & 
        \fix
     &
       5 & $6\times 6$, positive, \fix                     
     &
       $\hskip -1.65em
       f\in [0,1]^6$, \inp
                           & $>\lambda$ or $\ge \lambda$
    \\
    \hline\hline       
    Claus \cite{claus81}
           & $\pi= e_2$ &   9 (5\rlap) & $9\times 9$, posit{ive}, \inp
                                                      & $\eeta=e_1$
                           & $>1/9$ 
    \\
    Claus \cite{claus81}
           & $\pi= e_2$ & 2 &  $65\times 65$, posit{ive}, \inp & $\eeta=e_1$
                           & $>1/65$ 
    \\
           & $\pi= e_2$ & 2 &  \llap($37\times 37$), posit{ive}, \inp & $\eeta=e_1$
                           & $>1/37$ 
    \\
    B.\&C.~\cite{blondel-canterini-03}
           & $\pi= e_2$
             &2  & $46\times 46$, 
                             posit{ive}, \inp
                                                & 
                                                  $\eeta=e_1$
                                                        &  $>1/46$ or $\ge1/46$ 
    \\
           & $\pi= e_2$
             &2  & \llap($34\times 34$), 
                             posit{ive}, \inp
                                                & 
                                                  $\eeta=e_1$
                                                        &  $>1/34$ or $\ge1/34$ 
    \\
    Hirv.~
 \cite{hirvensalo07}           & $\pi=e_2$ & 2 & $25\times 25$, 
                             posit{ive}, \inp
  & $\eeta=e_1$ &  $>1/25$ 
\\   
           & $\pi=e_2$ & 2 & 
                             \llap($20\times 20$),
                             posit{ive}, \inp
                                                      & $\eeta=e_1$ &  $>1/20$ 
    \\
\hline       
  \end{tabular}
  
  \caption{The main characteristics of the data $\pi$,
    $\mathcal{M}$, and $\eeta$ for 
 different
 undecidable versions of PFA Emptiness.
    The vectors
    $e_1$ and $e_2$ are two standard unit vectors of appropriate
    dimension,
    indicating a single accepting state or a single deterministic
    start state.
    The results under the double line are previous results from
    the literature.
The references are to Claus
 \cite[Theorem 6(iii)]{claus81} from 1981,
  Blondel and Canterini
  \cite[Theorem~2.1]{blondel-canterini-03} from 2003, and Hirvensalo
  \cite[Section~3]{hirvensalo07}
  from 2007. 
  The numbers in parentheses 
  are the figures that
  would be obtained by basing the proofs on Neary's undecidable 5-word
instances of the Post Correspondence Problem (PCP)
    instead of the smallest known
    undecidable PCP instances that were current
    at the time\ifanonymous.
\else, see
\cite[Theorems~5 and~6]{r-pfaeu-24}. 
\fi   
  }
  \label{tab:results}
\end{table}

\subsection{Statement of results}
\label{sec:results}
There are three different 
proofs of
the basic undecidability
theorem
 in the literature.
The first proof is due to
Masakazu Nasu and Namio Honda~%
\cite{NasuHonda1969} from 1969.\footnote
{This proof is commonly
  misattributed to Paz~\cite{paz71},
  although Paz gave 
  credit
to Nasu and Honda
\cite[Section IIIB.7, Bibliographical notes, p.~193]{paz71}.}
The second proof, by Volker Claus \cite{claus81}, is loosely related to
this proof.
All these proves use a reduction from the undecidability of the
Post Correspondence Problem (PCP,
see \autoref{sec:PCP-intro}).
A completely independent proof, which
 uses a very different approach,
was given by
Anne Condon and Richard J. Lipton~%
\cite{condon-lipton-1989} in 1989,
based on ideas of
R{\={u}}si{\c{n}}{\v{s}}
Freivalds~\cite{freivalds81} from 1981.
\ifanonymous\else
The somewhat intricate 
history is 
described
in \cite[Section~3]{r-pfaeu-24} 
 as part of an extensive survey of the various undecidability proofs.
\fi


PFA Emptiness remains undecidable under various constraints on the
number of transition matrices (size of the alphabet) and their
size (number of states).
The first result in this direction is due to
 Claus
 \cite{claus81} from 1981. Later improvements were made
 by
  Blondel and Canterini
  \cite{blondel-canterini-03} and Hirvensalo
  \cite{hirvensalo07}, who were apparently unaware of~\cite{claus81}
and concentrated on the case of
  two matrices (\emph{binary} input alphabet).
An overview of these results is shown in \autoref{tab:results} below the double line.
  
We improve these bounds, concentrating on 
 the minimum number of states without restricting the
 number of matrices,
see \autoref{thm-matrices-forward}.
(Such results are only implicit in the proofs of \cite{blondel-canterini-03} and
\cite{hirvensalo07}.)
Undecidability even holds for a PFA where all data except the
starting distribution $\pi$ are fixed.

We also get improved bounds for the 2-matrix case.
For the variation where the output vector $f$ is not
restricted to a 0-1-vector, we can reduce the number of
states to~6
(\autoref{thm-matrices-backward}).

Our results are stated in
the following theorems and summarized in
\autoref{tab:results}.

\begin{theorem}
  \label{thm-matrices-forward}
The PFA Emptiness Problem 
 is undecidable for PFAs with a
 single accepting state and the following restrictions on the number
 of transition matrices (size of the input alphabet)
 and their size (number of states):
  \begin{enumerate}[\rm(a)]
  \item \label{4-positive}
    $\mathcal{M}$ consists of
    4 
    positive doubly-stochastic transition matrices of size $7\times 7$, with
    cutpoint $\lambda=1/7$.
  \item \label{5-fixed}
    $\mathcal{M}$ consists of
    6 
    fixed positive doubly-stochastic transition matrices of size $7\times 7$, with
    cutpoint $\lambda=1/7$.
    The only variable input is the starting distribution $\pi\in [0,1]^7$.
  \item  \label{2-matrices}
    $\mathcal{M}$ consists of
    2 
    positive transition matrices of size $18\times 18$, with cutpoint
    $\lambda=1/18$.
  \item  \label{2-matrices-fixed}
    $\mathcal{M}$ consists of
    2 
    fixed
    positive transition matrices of size $28\times 28$, with cutpoint
    $\lambda=1/28$.
\end{enumerate}   
All statements hold also for weak inequality $(\ge\lambda)$ as the  acceptance criterion.
\end{theorem}

\begin{theorem}
  \label{thm-matrices-backward}
For any cutpoint $\lambda$ in the interval $0<\lambda<1$,  
the PFA Emptiness Problem 
with an output vector $\eeta$ with entries $0\le\eeta_i\le1$
is undecidable for PFAs with
 the following restrictions on the number
 of transition matrices (size of the input alphabet)
 and their size (number of states):
  \begin{enumerate}[\rm(a)]
  \item \label{backward}
    $\mathcal{M}$ consists of
    4 
    positive
 transition matrices of size $6\times 6$.
   There is a fixed deterministic start state.
  \item \label{backward-fixed}
    $\mathcal{M}$ consists of
    5 
    fixed
 positive 
    transition matrices of size $6\times 6$.
   There is a fixed starting distribution~$\pi$.
    The only input of the problem
   is the vector $\eeta\in [0,1]^6$ of output values.
\end{enumerate}
All statements hold also for weak inequality $(\ge\lambda)$ as the  acceptance criterion.
\end{theorem}

By combining the different proof techniques, one can extend
\autoref{thm-matrices-backward} to PFAs with \emph{two matrices}
and a fixed start state or starting distribution,
analogous to
\autoref{thm-matrices-forward}\ref{2-matrices}--\ref{2-matrices-fixed}, but
we have not worked out these results.

%

\ifanonymous\else
Some weaker results with fixed matrices were previously obtained
in a technical report \cite[Theorems~2 and~4]{r-pfaeu-24}. 
For example, PFA Emptiness was shown to be undecidable for 52 fixed $9\times9$
matrices,
in the setting
of \autoref{thm-matrices-backward}
where the final vector $f$ is the only input,
with acceptance criterion~$\ge1/2$.
For acceptance with strict inequality ($>1/4$),
52 fixed
matrices of size  $11\times11$ were needed.
In these constructions, the PCP was derived from a universal Turing machine.
For the case of 2 fixed matrices with
a single accepting state, where the start distribution $\pi$ is the only  input, a bound of 572
states was obtained
\cite[Theorem~3]{r-pfaeu-24}. 

\fi

\ifnotkurz

\subsection{Contributions of this paper and relations to other problems}
\label{sec:contributions}
The most striking result of the paper is perhaps that
PFA Emptiness is  undecidable already for a \emph{fixed} PFA,
where only the starting distribution~$\pi$, or the output vector $f$
is an input.
This result is not deep; it follows quite readily from the combination of known
ideas:
The main observation, namely that the
reduction of
 Matiyasevich and Sénizergues~\cite{matiyasevich2005decision}
 from undecidable semi-Thue systems leads to instances of
 the Post Correspondence problem (PCP) where only the starting pair is variable,
 has been made by Halava, Harju, and Hirvensalo in 2007
 \cite[Theorem~6]
 {halava07-Claus-instances}.
The idea of merging the first or last matrix into the starting
distribution $\pi$
or into the final vector $f$ was used by Hirvensalo in
2007~\cite[Step~2 in Section~3]{hirvensalo07}.

The major technical contribution is the reduction of the number of
states to six, even for PFAs, and even while keeping the matrices
fixed,
at least for the version where the input is the (fractional) output vector~$f$.

\ifnotkurz
Weighted automata
are in some ways more natural than PFAs, and these
automata, in particular rational-weighted automata,
have recently attracted a lot of attention.
Since they generalize PFAs, all our 
undecidability results immediately carry over to rational-weighted
automata.

The PFA Emptiness Problem is a problem about matrix products.
There is a great variety of problems about matrix products
whose undecidability has been studied;
see~\cite{Advances-Semigroups-2023} for a recent survey of this rich
area.
In fact the first problem
outside
the 
fields of
mathematical logic and 
theory of computing
 that was shown to be undecidable is
a problem about matrix products:
Markov proved in 1947
that it is undecidable whether two semigroups
generated two finite sets of matrices contain
a common element~\cite{Markov-1947};
see Halava and
Harju~\cite{halava-harju-07_markov_undec_theor_integ_matric}
for a presentation of Markov's proof (and for an alternative proof).
Our basic approach is the same as Markov's: to model
the Post Correspondence Problem (PCP) by matrix products.
Our particular technique for doing this has its
roots in
 Paterson's undecidability proof~\cite{paterson70}
from 1970
of the \emph{matrix mortality problem} for $3\times 3$
 matrices. 



% ...



\fi

Besides, 
we have brought to the light some papers that were apparently
forgotten, like
Nasu and Honda's original undecidability proof from 1969~\cite{NasuHonda1969},
and Claus~\cite{claus81}.
Also, our technique for converting matrices with row sums~0 to positive matrices
with row sums~1 (hence stochastic matrices) in Section
\ref{sec:step4} 
is more streamlined and elegant than the proofs that we have seen in
the literature.

\fi


\subsection{The Post Correspondence Problem (PCP)}
\label{sec:PCP-intro}

In the \emph{Post Correspondence Problem} (PCP),
we are given a list of pairs of words $(v_1,w_1),
\allowbreak
(v_2,w_2), \ldots,
(v_k,w_k)$. 
 The problem is to decide if there is a nonempty sequence
$i_1i_2\ldots i_m$
of indices $i_j \in \{1,2,\ldots,k\}$ such that
\begin{displaymath}
  v_{i_1}v_{i_2}\ldots v_{i_m}
=  w_{i_1}w_{i_2}\ldots w_{i_m}
\end{displaymath}
This is one of the well-known undecidable problems.
According to Neary~\cite{neary:PCP5:2015},
 the PCP is already undecidable with as few as five
 word pairs.

\section{Proofs of \autoref{thm-matrices-forward}\ref{4-positive}
  and  \autoref{thm-matrices-forward}\ref{5-fixed} (few states)}
\label{sec:more-than-two}

  We follow the same overall proof strategy as
Claus
 \cite{claus81}, 
  Blondel and Canterini~\cite{blondel-canterini-03} and Hirvensalo~\cite{hirvensalo07}:
  They 
  use undecidable PCP instances with few word pairs and transform them
  to the Emptiness Problem for \emph{integer-weighted} automata, which are
  then
  converted to PFAs.
 We deviate from this approach by using an automaton with \emph{fractional}
 weights  (\autoref{sec:step1}).
These matrices can be converted to column-stochastic matrices without
the overhead of an extra state (\autoref{sec:step2}).

 The proof of \autoref{thm-matrices-forward}\ref{4-positive} contains
 the main ideas.
For the reduction to two matrices, we apply a technique of 
 Hirvensalo~\cite{hirvensalo07} (\autoref{sec:reduce-to-2}).
 All other results are obtained by slight variations of these methods
in combination with appropriate results from the literature.

\subsection{Step 1: Modeling word pairs by matrices}
\label{sec:step1}
For a string $u=u_1u_2\ldots u_n$ of decimal digits,
we denote its fractional 
 decimal value by
$0.u=
\sum_{j=1}^n u_j\cdot 10^{-j}$. For example, if $u=\mathtt{432100}$,
then
$0.u=0.4321$.
We will take care to avoid 
trailing zeros, because their disappearance could cause problems.

For two strings $v,w$ of digits in $\{\mathtt{11},\mathtt{12}\}^*$, we define the matrix
\begin{displaymath}
  A_0(v,w) =
\begin{pmatrix}
1 & 0 & 0 & 0 & 0 & 0 \\
0.v & 10^{-|v|} & 0 & 0 & 0 & 0 \\
(0.v)^2 & 2 \cdot 10^{-|v|}\cdot 0.v & 10^{-2|v|} & 0 & 0 & 0 \\
0.w & 0 & 0 & 10^{-|w|} & 0 & 0 \\
(0.w)^2 & 0 & 0 & 2 \cdot 10^{-|w|}\cdot 0.w & 10^{-2|w|} & 0 \\
0.v\cdot0.w & 10^{-|v|}\cdot 0.w & 0 & 10^{-|w|}\cdot 0.v & 0 & 10^{-|v|-|w|}
\end{pmatrix}.
\end{displaymath}
It is not straightforward to see, but it can be checked by a
simple 
calculation 
that the matrices $A_0(v,w)$ satisfy a
multiplicative law
(see \Autoref{sec:check} for a computer check):

\begin{lemma}[Multiplicative Law]
  \label{lem:multi}
\begin{equation}
  \label{eq:multi_0}
 A_0 (v_1,w_1)
 A_0 (v_2,w_2) =
 A_0 (v_1v_2,w_1w_2)
\end{equation}
\end{lemma}

Now we transform these matrices
$A_0 (v,w)$
by a similarity transformation with the transformation matrix
\begin{equation}\label{eq:transform-U}
U= \begin{pmatrix}
1 & 0 & 0 & 0 & 0 & 0 \\
0 & 1 & 0 & 0 & 0 & 0 \\
\frac{1}{99} & 0 & 1 & 0 & 0 & 0 \\
0 & 0 & 0 & 1 & 0 & 0 \\
0 & 0 & 0 & 0 & \frac{99}{105} & 0 \\
0 & 0 & 0 & 0 & 0 & 1
\end{pmatrix},\quad
U^{-1}= \begin{pmatrix}
1 & 0 & 0 & 0 & 0 & 0 \\
0 & 1 & 0 & 0 & 0 & 0 \\
-\frac{1}{99} & 0 & 1 & 0 & 0 & 0 \\
0 & 0 & 0 & 1 & 0 & 0 \\
0 & 0 & 0 & 0 & \frac{105}{99} & 0 \\
0 & 0 & 0 & 0 & 0 & 1
\end{pmatrix}.
\end{equation}
The resulting matrices
$A (v,w) = U^{-1} A_0 (v,w) U$
differ in some entries of the first column and the fifth row:
\begin{displaymath}
  \begin{pmatrix}
1 & 0 & 0 & 0 & 0 & 0 \\
0.v & 10^{-|v|} & 0 & 0 & 0 & 0 \\
(0.v)^2 +
\frac{1}{99} (1{-}10^{-2|v|})
& 2 \cdot 10^{-|v|}\cdot 0.v & 10^{-2|v|} & 0 & 0 & 0 \\
0.w & 0 & 0 & 10^{-|w|} & 0 & 0 \\
\frac{99}{105}\cdot
(0.w)^2 & 0 & 0 &\!\!\!\!
\frac{198}{105} \cdot 10^{-|w|}\cdot 0.w & 10^{-2|w|} & 0 \\
0.v\cdot0.w & 10^{-|v|}\cdot 0.w & 0 & 10^{-|w|}\cdot 0.v & 0 &
\!\!\!\!
10^{-|v|-|w|}
                \end{pmatrix}
              \end{displaymath}
Since the matrices were obtained by a similarity, they satisfy the same multiplicative law:
\begin{equation}
  \label{eq:multi}
 A (v_1,w_1)
 A (v_2,w_2) =
 A (v_1v_2,w_1w_2).
\end{equation}         
With the vectors
$\pi_1=(\frac{1}{99},0,-1,0,-\frac{105}{99},2)$ and $f_1=(1,0,0,0,0,0)^T$,
we obtain
\begin{align}\nonumber
  \pi_1 A(v,w) f_1
 & = \tfrac{1}{99}
   - (0.v)^2 -\tfrac{1}{99} +10^{-2|v|}/99 -(0.w)^2
   +
   2\cdot 0.v\cdot0.w 
\\ & =
     - (0.v-0.w)^2 +10^{-2|v|}/99
     \label{eq:almost-zero}
\end{align}
The sign of this expression can detect inequality of $v$ and $w$, as
we
will presently show.

We could have taken the simpler matrix $A_0$ 
with $\pi_0=(0,0,-1,0,-1,2)$, and this would give
$\pi_0 A_0(v,w) f_1 = - (0.v-0.w)^2$. The contortions with the matrix
$U$
were necessary to generate a tiny positive term in \eqref{eq:almost-zero}.
The reason for the peculiar entries $\frac{99}{105}$ and $\frac{105}{99}$
in \eqref{eq:transform-U} and in $\pi_1$
will become
apparent after the next transformation.

\subsection{Equality detection}
\label{sec:equality}

\begin{lemma}[Equality Detection Lemma] \label{lemma:modif-threshold}
  Let $v,w \in \{\mathtt{11},\mathtt{12}\}^*$.
  \begin{enumerate}
  \item       If $v=w$,
    then
    $(0.v-0.w)^2-10^{-2|v|}/99 < 0$.
  \item       If $v\ne w$,
    then
    $(0.v-0.w)^2-10^{-2|v|}/99 > 0$.
  \end{enumerate}
 In particular, $(0.v-0.w)^2-10^{-2|v|}/99$ is never zero.
\end{lemma}
\begin{proof}
  The first statement is obvious.

To see the second statement,
we first consider the case that one of the strings
is a  prefix of
the other 
(Case A):
If $w$ is a prefix of $v$, then for
fixed $w$,
  the smallest possible difference
$|0.v-0.w|$ among the strings $v$ that extend $w$
is achieved 
when
$0.v=0.w\mathtt{11}$.
Similarly,
if $v$ is a prefix of $w$, then
  the smallest possible difference
is achieved 
when
$0.w=0.v\mathtt{11}$.
%
In either case,
$|0.v-0.w| \ge 10^{-\min\{|v|,|w|\}}\cdot 0.11
\ge
10^{-|v|} \cdot 0.11
$.
Thus,
$(0.v-0.w) ^2
\ge
10^{-2|v|} \cdot 0.0121
>10^{-2|v|} \cdot 0.010101\ldots
=10^{-2|v|}/99
$.


Case B: Neither of $v$ and $w$  is a prefix of the other.
  Suppose
  $v$ and $w$
share $k$ leading digit pairs  $u\in \{\mathtt{11},\mathtt{12}\}^k$,
$0\le k<  |v|/2$.
Then one of the two strings
 starts with
$u\mathtt{12}$ and the other with 
$u\mathtt{11}$;
the smallest difference
$|0.v-0.w|$
between
two such numbers
is achieved between 
$0.u\mathtt{12}$ and
$0.u\mathtt{11121212\ldots}$, 
and thus
$|0.v-0.w| > 10^{-2k}\cdot 0.00878787\ldots
> 10^{-2k}\cdot 0.005\ge
10^{-|v|+2}\cdot 0.005=
10^{-|v|}/2$.
After squaring this relation,
the claim follows with an ample margin.
\end{proof}

\subsection{Modeling the Post Correspondence Problem (PCP)}
\label{sec:PCP}

The multiplicative law
 \eqref{eq:multi}
and the capacity to detect string equality are
all that is needed to model the PCP.

In the PCP,
it is no loss of generality to assume that the words in the pairs $(v_i,w_i)
\allowbreak
$
use a binary alphabet, since any alphabet can be coded in binary.
We recode them to the binary ``alphabet''
$\{\mathtt{11},\mathtt{12}\}$,
i.e., they become
words
in $\{\mathtt{11},\mathtt{12}\}^*$
of doubled length.
We form the matrices $A_i=A(v_i,w_i)$.
By multiplicativity,
$A_{i_1}A_{i_2} \ldots A_{i_m} =
A(v_{i_1}v_{i_2} \ldots v_{i_m},
w_{i_1}w_{i_2} \ldots w_{i_m})$, and
by 
\autoref{lemma:modif-threshold},
$  \pi_1A_{i_1}A_{i_2} \ldots A_{i_m}f_1 >0$ if and only if
$v_{i_1}v_{i_2} \ldots v_{i_m}=
w_{i_1}w_{i_2} \ldots w_{i_m}$, i.\,e.,
$i_1i_2 \ldots i_m$ is a solution of the PCP.

Since the value 0 is excluded by
\autoref{lemma:modif-threshold},
nothing changes if the condition ``$>0$'' is replaced by
``$\ge 0$.''
This property will propagate through the proof,
with the consequence that in the resulting theorems,
it does not matter if we
take $> \lambda$ or $\ge\lambda$ as the acceptance criterion
for the PFA. We will not mention this issue any more and work
only with strict inequality.

There are two problems that we still have to solve:
\begin{itemize}
\item The empty sequence ($m=0$) always fulfills the inequality
although it does not count as a solution of the PCP.
\item 
The matrices $A_i$ are not stochastic, and $\pi_1$ is not a
probability distribution. So far, what we have is
a \emph{generalized probabilistic automaton}, or
\emph{rational-weighted automaton}.
\end{itemize}
Turakainen~\cite{turakainen69} showed in 
1969 that a
generalized probabilistic automaton can be converted to a PFA without changing the
 recognized language (with the possible exception of the empty word), and
 this can even be done by adding only two more states
\cite[Theorem~1(i)] {Turakainen_1975}.
Thus, by the general technique of 
\cite{Turakainen_1975}, we can get a PFA with 8 states.
We will use some
tailored version of these techniques,
involving nontrivial techniques,
so that no extra state is needed to make the matrices
stochastic and to simultaneously get rid of the empty solution.

\paragraph{History of ideas.}
The idea of
modeling of the PCP by multiplication of integer matrices
was
pioneered by 
Markov~\cite{Markov-1947} in 1947. He
used the matrices
$(
\begin{smallmatrix}
  1&0\\1&1
\end{smallmatrix}
)$
and $(
\begin{smallmatrix}
  1&1\\0&1
\end{smallmatrix}
)$,
which generate a free semigroup.
A different representation, closer to the one we are using,
goes back to
Paterson~\cite{paterson70}
in 1970, who used it to show that mortality for $3\times 3$
 matrices is undecidable. 

Our matrix $A_0(v,w)$ is a variation of the integer matrices that were
proposed
by Volker Claus \cite{claus81} in 1981,
and again by Blondel and Canterini \cite[p.~235]{blondel-canterini-03} in
2003,
in the very context of constructing small
undecidable instances of the PFA Emptiness Problem.
These matrices extend Paterson's matrices to
larger matrices that can produce quadratic terms.
They 
use \emph{positive} powers of the base 10
and \emph{integer} decimal values $(u)_{10}$ of the strings $u$ instead of fractional
values $0.u$.
Such matrices satisfy a multiplicative law such as~\eqref{eq:multi_0} but
with a reversal of
the factors.
In fact, the method works equally well for any other radix instead of~10.\footnote{ ``The notation is quaternary, decimal, etc., according
  to taste.'' (Paterson~\cite{paterson70}).
Claus
\cite{claus81} 
used radix~3. 
}

Our novel idea is to use this construction
with \emph{negative} powers of 10, leading to something that is roughly at
the same scale as a stochastic matrix, thus facilitating further
transformations.


\subsection{Step 2: Making column sums 1}
\label{sec:step2}
 We apply another similarity transformation, using the matrix
\begin{displaymath}
V= \begin{pmatrix}
1 & 1 & 1 & 1 & 1 & 1 \\
0 & 1 & 0 & 0 & 0 & 0 \\
0 & 0 & 1 & 0 & 0 & 0 \\
0 & 0 & 0 & 1 & 0 & 0 \\
0 & 0 & 0 & 0 & 1 & 0 \\
0 & 0 & 0 & 0 & 0 & 1
\end{pmatrix}, \
V^{-1}= \begin{pmatrix}
1 & -1 & -1 & -1 & -1 & -1 \\
0 & 1 & 0 & 0 & 0 & 0 \\
0 & 0 & 1 & 0 & 0 & 0 \\
0 & 0 & 0 & 1 & 0 & 0 \\
0 & 0 & 0 & 0 & 1 & 0 \\
0 & 0 & 0 & 0 & 0 & 1
\end{pmatrix},
\end{displaymath}
but this time, we also transform the vectors $\pi_1$ and $f_1$,
leaving the overall result 
unchanged:
\begin{align}\nonumber
  \pi_1A_{i_1}A_{i_2} \ldots A_{i_m}f_1 &=
  (\pi_1V)
  (V^{-1}A_{i_1}V)
  (V^{-1}A_{i_2}V) \ldots
  (V^{-1}A_{i_m}V) 
  (V^{-1}f_1) \\&=
  \pi_2B_{i_1}B_{i_2} \ldots B_{i_m}f_2
  \label{eq:B}
\end{align}
We analyze the effect of the similarity transform for a matrix
of the general form
\begin{displaymath}
A= \begin{pmatrix}
1 & 0 & 0 & 0 & 0 & 0 \\
a_{21} & a_{22} & 0 & 0 & 0 & 0 \\
a_{31} & a_{32} & a_{33} & 0 & 0 & 0 \\
a_{41} & a_{42} & a_{43} & a_{44} & 0 & 0 \\
a_{51} & a_{52} & a_{53} & a_{54} & a_{55} & 0 \\
a_{61} & a_{62} & a_{63} & a_{64} & a_{65} & a_{66}
\end{pmatrix} .
\end{displaymath}
The transformed matrix is
\begin{multline*}
B =  V^{-1}AV= \\\ \ \left(\!\!
\begin{array}{l@{\ \ \ }l@{\ \ \ }l@{\ \ }l@{\ \ }l@{\ \ }l@{\,}}
  K &
 K_{12} 
  & K {-} a_{33} {-} a_{43} {-} a_{53} {-} a_{63}
  & K {-} a_{44} {-} a_{54} {-} a_{64}
  & K {-} a_{55} {-} a_{65}
  & K - a_{66} \\
a_{21} & a_{21} + a_{22} & a_{21} & a_{21} & a_{21} & a_{21} \\
a_{31} & a_{31} + a_{32} & a_{31} + a_{33} & a_{31} & a_{31} & a_{31} \\
a_{41} & a_{41} + a_{42} & a_{41} + a_{43} & a_{41} + a_{44} & a_{41} & a_{41} \\
a_{51} & a_{51} + a_{52} & a_{51} + a_{53} & a_{51} + a_{54} & a_{51} + a_{55} & a_{51} \\
a_{61} & a_{61} + a_{62} & a_{61} + a_{63} & a_{61} + a_{64} & a_{61} + a_{65} & a_{61} + a_{66}
\end{array}\right) \!\!\!\!\!\!
\end{multline*}
with
$K= 1-a_{21} - a_{31} - a_{41} - a_{51} - a_{61}$
and $K_{12} =  K {-} a_{22} {-} a_{32} {-} a_{42} {-} a_{52} {-} a_{62}$.

\begin{lemma}
  \label{lem:column-stochastic}
  \begin{enumerate}
  \item the matrix $B=V^{-1}AV$ has column sums
$1$.
\item For $v\ne\epsilon$ and
 $w\ne\epsilon$, the matrix $V^{-1}A(v,w)V$ is
 positive,
and therefore column-stochastic.
\end{enumerate}

\end{lemma}
\begin{proof}
  The first statement can be easily checked directly,
  but it has a systematic reason:
\begin{displaymath}
  (1,1,1,1,1,1)V^{-1}AV=
  (1,0,0,0,0,0)AV=
  (1,0,0,0,0,0)V=
  (1,1,1,1,1,1)
\end{displaymath}

The second statement is not needed for \autoref{thm-matrices-forward}, because positivity
is established anyway in \hyperref[sec:step4]{Step 4}, after destroying it in \hyperref[sec:step3]{Step~3}, but we
need it for \autoref{thm-matrices-backward}.
So let us check it: The only entries that are in danger of becoming negative
are the entries of the first row.
The two entries $a_{21}=0.v$ and  $a_{41}=0.w$ of the matrix $A(v,w)$ can be as large as
$0.121212\ldots \le 0.15$; all other entries are safely below $0.05$.
Thus, even the ``most dangerous'' candidate $K_{12}$, where 10 of the entries
$a_{ij}$ are subtracted from 1, cannot be zero or negative.

It can be checked that rows 2--6 are positive, because the
first column of $A(v,w)$, which is positive, has been added to every other column.
\end{proof}

In the transformation from $\pi_1$ to
$\pi_2 = \pi_1 V$, the first entry is added to all other entries.
Thus, the vector
$\pi_1=(\frac{1}{99}
,
0,-1,0,-1-\frac{6}{99},2)$, 
becomes 
$\pi_2=(\frac{1}{99},\frac{1}{99},-1+\frac{1}{99},\frac{1}{99},-1-\frac{5}{99},2+\frac{1}{99})$,
whose entries sum to zero:
\begin{equation}
  \label{eq:pi-sum-zero}
  \pi_2\left(\small 
  \begin{matrix}
    1\\1\\1\\1\\1\\1
  \end{matrix}\right)
=
  \pi_1V\left(\small
  \begin{matrix}
    1\\1\\1\\1\\1\\1
  \end{matrix}\right)
=
  \pi_1\left(\small
  \begin{matrix}
    6\\1\\1\\1\\1\\1
  \end{matrix}\right)
= 0
\end{equation}
It was for this reason that the entry
 $\frac{105}{99}$ of $\pi_1$ and the corresponding entries of
 the transformations \eqref{eq:transform-U} were chosen.
 
The output vector $f$ is unchanged by the transformation:
$f_2 =  V^{-1}f_1=f_1$.

\subsection{Step 3: Making row sums 1 with
  an extra state}
\label{sec:row-1}
\label{sec:step3}
By adding an extra column $r_i$, we now make all row sums equal to 1.
We also add an extra row, resulting in a $7\times7$ matrix
\begin{equation}
  C_i =
  \begin{pmatrix}
    B_i & r_i\\
    0 & 1
  \end{pmatrix}
\end{equation}
The entries of $r_i$ may be negative. They lie in the range
$-5\le r \le 1$.
All column sums are still 1: Since the row sums are 1, the total sum
of the entries is 7. Therefore, since the first six column sums are 1, the sum of
the last column must also be 1.

The vectors $\pi_2$ and $f_2$
are extended to vectors
$\pi_3$ and $f_3$
of length 7 by adding a 0, but
they are otherwise unchanged.  Thus, the extra column and row of $C_i$
plays no role for the value of the product~\eqref{eq:B}, which remains
unchanged:
\begin{displaymath}
  \pi_3C_{i_1}C_{i_2} \ldots C_{i_m}f_3
=    \pi_2B_{i_1}B_{i_2} \ldots B_{i_m}f_2
\end{displaymath}
\subsection{Step 4: Making the matrices positive, and hence
  stochastic}
\label{sec:step4}

Let $J$ be the doubly-stochastic $7\times 7$ transition matrix of
the ``completely random transition'' with all entries $1/7$.
Then, with 
$\alpha=0.01$,
we form the matrices $D_i := (1-\alpha)J + \alpha C_i$.
The constant $\alpha$ is small enough to ensure that $D_i>0$.
Hence the matrices $D_i$ are doubly-stochastic.

If expand of the product
$\pi_3\prod_{j=1}^m D_{i_j}=
\pi_3\prod_{j=1}^m \bigl((1-\alpha)J + \alpha C_{i_j}\bigr)$,
we get a sum of $2^m$ terms, each containing a
 product of $m$ of the matrices $J$ or $C_i$.
We find that all terms containing the factor $J$ vanish. 
The reason is that, 
since $C_i$ has row and column sums~1,
 $JC_i = C_iJ=J$. Moreover, $\pi_3 J=0$ by \eqref{eq:pi-sum-zero}.
It follows that
\begin{displaymath}
    \pi_3D_{i_1}D_{i_2} \ldots D_{i_m}f_3
=  \alpha^m \cdot \pi_3C_{i_1}C_{i_2} \ldots C_{i_m}f_3
\end{displaymath}
The factor $\alpha^m$ plays no role for the sign, and hence,
\begin{equation}
  \label{eq:D}
  \pi_3D_{i_1}D_{i_2} \ldots D_{i_m}f_3 >0  
\end{equation}
if and only if
$i_1i_2 \ldots i_m$ is a solution of the PCP.

\subsection{Step 5: Turning \texorpdfstring{$\pi$}{π} into a
  probability distribution}
\label{sec:step5}
The start vector $\pi_3$ has sum 0 and contains negative entries,
which
are smaller than $-1$ but not smaller than $-2$.
We form $\pi_4 = \bigl( (2,2,2,2,2,2,2)+\pi_3) \bigr)/14$.
It is positive and has sum~1.
The effect of substituting $\pi_3$ by $\pi_4$ in \eqref{eq:D} is that
the result is increased by $1/7$.
The reason is that
\begin{math}
  (1,1,1,1,1,1,1)
  D_{i_1}D_{i_2} \ldots D_{i_m} =   (1,1,1,1,1,1,1)
\end{math},
since the matrices $D_i$ are doubly-stochastic, and in particular,
column-stochastic,
and therefore 
\begin{math}
  (2,2,2,2,2,2,2)/14\cdot
  D_{i_1}D_{i_2} \ldots D_{i_m}f_3 = 2/14 = 1/7
\end{math}.

In summary, we have now constructed a true PFA with seven states that models the PCP:
\begin{equation}
  \label{eq:D2}
  \pi_4D_{i_1}D_{i_2} \ldots D_{i_m}f_3 >1/7
\end{equation}
if and only if
$i_1i_2 \ldots i_m$ is a (possibly empty) solution of the PCP.

\subsection{Using a small PCP}
\label{sec:select-PCP}

We base our proof on the undecidability of the PCP with 5 word pairs, as
established by Turlough Neary~\cite[Theorem~11]{neary:PCP5:2015} in 2015.
Neary constructed PCP instances with five pairs
$(v_1,w_1),\,(v_2,w_2),\,(v_3,w_3),\,(v_4,w_4),\,(v_5,w_5)$
that have the following property:
Every solution necessarily starts with the pair
$(v_1,w_1)$ and ends with
$(v_5,w_5)$.
We can also assume that
the end pair $(v_5,w_5)$ is used nowhere else.
(More precisely, in every \emph{primitive} solution (which is not a
concatenation of smaller solutions),
the end pair $(v_5,w_5)$ \emph{cannot} appear anywhere else than in
the final position.)
However, the
start pair
$(v_1,w_1)$ is also used in the middle of the solutions.
(This multipurpose usage of the
start pair is one of the devices to achieve such a remarkably small number of
word pairs.)\footnote
{The proof of Theorem~11
in Neary~\cite{neary:PCP5:2015}
contains an error, but this error can be fixed:
His PCP instances encode \emph{binary tag systems}.
 When showing that the PCP solution must
 follow the intended patters of the simulation of the binary tag system,
 Neary \cite[p.~660]{neary:PCP5:2015}
needs to show that the end pair $(v_5,w_5)=
(10^\beta1111,1111)$ cannot be used except to bring the two strings to
a common end. He
 claims that
 a block $1111$ cannot appear in the encoded string because in $u$
(the unencoded string of the binary tag system, which is described in Lemma~9)
 we
cannot have two
$c$ symbols
next to each other.
This is not true. The paper contains plenty of examples,
and they
contradict this claim; for example, the string $u'$ in (7)
\cite[p.~657]{neary:PCP5:2015} contains seven $c$'s in a row.
The mistake can be fixed by taking a longer block of 1s:
Looking at the appendants in Lemma~9, it is clear that every block of
length $|u|+1$ must contain a symbol $b$. Thus, the 
pair
$(v_5,w_5)=
(10^\beta1^{|u|+99},1^{|u|+99})$ will work as end pair.
}

\paragraph{Reversing the words.}
For our construction it is preferable to have the
word pair
$(v_5,w_5)$ at the beginning. We thus reverse all words.
Any solution sequence
$i_1i_2\ldots i_m$ for the original problem
must be also be reversed to
$i_mi_{m-1}\ldots i_1$, but this does not affect
the solvability of the PCP.
Thus, we work with PCP instances of the form
 $(v_1^R,w_1^R),
\allowbreak
(v_2^R,w_2^R), \ldots,
(v_5^R,w_5^R)$, with the following property:
Every
 solution sequence
 $i_1i_2\ldots i_m$
 must start with the pair
$(v_5^R,w_5^R)$,
and the pair
 $(v_5^R,w_5^R)$ cannot be used anywhere else:
 $i_1=5$ and $1\le i_j\le 4$ for $j=2,\ldots,m$.

\subsection{Step 6. Merging the leftmost matrix into the starting distribution} 
\label{sec:step6}
\label{sec:merge-left}
We apply the above construction steps to the word pairs $(v_1^R,w_1^R),
\allowbreak
\ldots,
(v_5^R,w_5^R)$, leading to five matrices $D_1,D_2,D_3,D_4,D_5$.
Since the leftmost matrix must be $D_5$ in any solution,
we can combine this 
matrix $D_5$ with the starting
distribution $\pi_4$: 
\begin{equation}
  \nonumber 
  \pi_4D_{5}D_{i_2} \ldots D_{i_m}f_3 =
  \pi_5D_{i_2} \ldots D_{i_m}f_3 ,
\end{equation}
leading to a new starting distribution
$  \pi_5 :=
\pi_4D_{5}$.
The matrix $D_5$ can be removed from the pool $\mathcal{M}$ of
matrices,
leaving only 4 matrices.
We have simultaneously eliminated the empty solution with a product of
$m=0$ matrices.
This concludes the proof of
\autoref{thm-matrices-forward}\ref{4-positive}.
\qed

\subsection{Proof of \autoref{thm-matrices-forward}\ref{5-fixed} (fixed
  transition matrices) by using a PCP with only one variable word pair}
\label{sec:Mati}

 Matiyasevich and S{\'e}nizergues \cite{matiyasevich2005decision}
 constructed PCP instances with seven word pairs
 $(v_1,w_1),\,\allowbreak(v_2,w_2),\,(v_3,w_3),\,(v_4,w_4),\,(v_5,w_5),
 \,(v_6,w_6),
 \,(v_7,w_7)$
that have the following property:
Every solution necessarily starts with the pair
$(v_1,w_1)$ and ends with
$(v_2,w_2)$.
Both the start pair $(v_1,w_1)$ and the end pair $(v_2,w_2)$ can
be assumed to appear nowhere else,
in the sense described in \autoref{sec:select-PCP}, i.\,e.,
apart from the possibility to concatenate solutions to obtain
longer solutions.
Matiyasevich and S{\'e}nizergues \cite
{matiyasevich2005decision}
used a reduction,
 due to
 Claus~\cite{\clausEATCS},
from
the \emph{individual accessibility problem} 
(or individual word problem)
for semi-Thue systems.
They showed that
the individual accessibility problem
is already undecidable for 
a \emph{particular} semi-Thue system with 3 rules~\cite[Theorem~3]{matiyasevich2005decision}.
Halava, Harju, and Hirvensalo
\cite[Theorem~6]{halava07-Claus-instances}
observed an important consequence of  this:
One can fix all words except $v_1$, and leave only
$v_1$ as an input to the problem, and the PCP is still undecidable.

From these word pairs,
we form the matrices
$A_i=A(v_i,w_i)$,
from the words
without reversal.
Following the same steps as above, we
eventually arrive at corresponding matrices $D_1,\ldots,D_7$.
 We merge
$D_1$ with $\pi_4$ into
a new starting distribution
$  \pi_5 :=
\pi_4D_{1}$.
We are left with a pool $\mathcal{M}=\{D_2,\ldots,D_7\}$ of
six fixed
matrices. The only variable input is the
starting distribution
$  \pi_5$.
This concludes the 
proof of \autoref{thm-matrices-forward}\ref{5-fixed}.
\qed

\section{Proofs of \autoref{thm-matrices-forward}\ref{2-matrices}
and \autoref{thm-matrices-forward}\ref{2-matrices-fixed}
  (binary input)}
\label{sec:two-matrices}

\subsection{Reduction to two matrices} 
\label{sec:reduce-to-2}

The following lemma and its proof is extracted from Step~3 in
Hirvensalo \cite[Section~3]{hirvensalo07}.
We denote by
$\phi(u)$ the acceptance probability of a (generalized) PFA for
an input $u$, i.e., the value of the product in the expression
\eqref{eq:accept}.

\begin{lemma}
\label{lem:reduce-to-2}
    Let $A=
(\pi,\mathcal{M},f)$
  be a generalized PFA with $k=|\Sigma|\ge3$ matrices $M_i$ of size
  $d\times d$, such that the first row of each matrix $M_i$ is
  $(1,0,\ldots,0)$.

  Then
    one can obtain
    a generalized PFA $A'=(\pi',\mathcal{M}',f')$
    over the two-symbol alphabet $\{\mathtt{a,b}\}$,
 with
    $2$ 
 matrices
 $M_{\mathtt{a}}'$ and
 $M_{\mathtt{b}}'$
 of dimension
 $(k-1)(d-1)+1$ such that
 for every word $u'
 \in  \{\mathtt{a,b}\}^*$ there exists a word $u\in\Sigma^*$ with
 \begin{equation}\label{eq:exists-a}
   \phi(u)=\phi'(u'),
 \end{equation}
and conversely, 
for every word
$u\in\Sigma^*$ there exists a word 
$u'
 \in  \{\mathtt{a,b}\}^*$ 
 with
 \eqref{eq:exists-a}.

 If the given matrices $M_i$ are nonnegative, then so are
$M_{\mathtt{a}}'$ and
$M_{\mathtt{b}}'$.
 If the given matrices $M_i$ are stochastic, then so are
$M_{\mathtt{a}}'$ and
$M_{\mathtt{b}}'$.
If $\pi$ is a probability distribution, then so is $\pi'$.
The entries of $f'$ are taken from the entries of $f$.

Thus, if $A$ is a PFA, then $B$ is a PFA as well.
 
\end{lemma}
The lemma is not specific about the word $u$ or $u'$ whose existence is
guaranteed.
Nevertheless, regardless of how these words are found,
the statement
is sufficient in the context of the emptiness
question.

However,
we can be explicit about $u$ and $u'$:
The construction uses a binary encoding
  $\tau\colon\Sigma^*\to
  \{\mathtt{a,b}\}^*
  $ with the prefix-free codewords
  $\{
 \mathtt{b,ab},\mathtt{aab},\allowbreak
 \ldots,\allowbreak\mathtt{a}^{k-2}\mathtt{b}
,\allowbreak\mathtt{a}^{k-1}
\}$.
Then, we can take $u'=\tau(u)$, because
 \begin{displaymath}
   \phi(u)=\phi'(\tau(u)).
 \end{displaymath}
 For words $u'$ that are not of the form $\tau(u)$,
 \eqref{eq:exists-a} holds for the longest word $u$ such that
  $\tau(u)$ is a prefix of $u'$. In other words, $u$ is the decoding
  of the longest decodable prefix of~$u'$.

\begin{proof}
  The procedure is easiest to describe if we assume that $A$ is a PFA.
  Then we can think of $A'$ as an automaton that decodes the input
 $u'\in\{\mathtt{a},\mathtt{b}\}^*$ and carries out a transition of~$A$ whenever it has
  read a complete codeword.  In addition to the state
  $q\in\{1,\ldots,d\}$ of $A$, $A'$ needs to maintain a counter $i$ in
  the range $0\le i \le k-2$ for the number of \texttt{a}'s of the
  current partial codeword.  Whenever $A'$ reads a \texttt{b}, or if it has read the
  $(k-1)$-st \texttt{a}, $A'$ performs the appropriate random
  transition and resets the counter.  The number of states of $A'$ is
  $d\times(k-1)$.
    
  The fact that state 1 is an absorbing state allows a shortcut:
  If we are in state 1, we can stop to maintain the counter $i$.
  Thus, only states $2,\ldots,d$ need to be multiplied by $k-1$, and
  the resulting number of states reduces to $(d-1)(k-1)+1$.

  We will now describe the construction of $A'$ in terms of transition
  matrices.
  This construction is valid also when $A$ is a generalized PFA, where
  the above description in terms of random transitions makes no sense. 
  To make the description more concrete, we illustrate it with
  $k=5$ matrices $M_1,\ldots,M_5$
and the corresponding binary  codewords
\texttt{b}, \texttt{ab}, \texttt{aab}, \texttt{aaab}, \texttt{aaaa}.

We split the $d\times d$ matrices $M_i$ and the vectors $\pi$ and $f$ into
blocks of size $1+(d-1)$:
\begin{displaymath}
  M_i =
  \begin{pmatrix}
    1&0\\
c_i&    \hat C_i \\
  \end{pmatrix}
,\ 
  \pi =
  \begin{pmatrix}
p_1\\    \hat \pi\\
\end{pmatrix}^{\!\!T}
,\
  \eeta =
  \begin{pmatrix}
    f_1\\
    \hat \eeta\\
  \end{pmatrix}
.
\end{displaymath}
We define new transition matrices and start and end vectors
in block form
with block sizes $1+(d-1)+(d-1)+(d-1)+(d-1)$
as follows:
\begin{displaymath}
  M'_{\mathtt{a}}=
  \begin{pmatrix}
  1&0&0&0&0\\
  0& 0 & I & 0&0 \\
  0& 0 & 0 & I&0  \\
 0& 0& 0 & 0 & I  \\
c_5 &  \hat C_5  &  0 &  0 & 0 \\
\end{pmatrix}
,\ 
M'_\mathtt{b}=
  \begin{pmatrix}
  1 & 0 & 0 & 0& 0\\
c_1&    \hat C_1 & 0 & 0& 0\\
c_2&    \hat C_2 & 0 & 0& 0\\
c_3&    \hat C_3 & 0 & 0& 0\\
c_4&    \hat C_4 & 0 & 0& 0\\
    \end{pmatrix}
,\ 
\pi' =
\begin{pmatrix}
  p_1\\
 \hat \pi_2\\0\\0\\0
\end{pmatrix} ^{\!\!T}
\!\!\text{, and }
f' =
\begin{pmatrix}
f_1\\\hat f\\\hat f\\\hat f\\\hat f
\end{pmatrix}
.
\end{displaymath}
From the sequence of powers of $M'_{\mathtt{a}}$,
\begin{displaymath}
  (M'_{\mathtt{a}})^2
=
  \begin{pmatrix}
1&  0&0&0&0\\
  0 & 0 & 0 & I &0\\
0&  0 & 0 & 0 & I \\
c_5&  \hat C_5  &  0 & 0 &0\\
c_5& 0&  \hat C_5  & 0 & 0\\
\end{pmatrix},
\  (M'_{\mathtt{a}})^3
=
  \begin{pmatrix}
1&  0&0&0&0\\
0&  0 & 0 & 0 & I \\
c_5&  \hat C_5  &  0 & 0 &0\\
c_5& 0&  \hat C_5  & 0 & 0\\
c_5& 0& 0& \hat C_5  & 0\\
\end{pmatrix},
\end{displaymath}
we can recognize the pattern of development. 
We can then work out the matrices
$M'_{\mathtt{b}}$,
$M'_{\mathtt{a}}
M'_{\mathtt{b}}$,
$M'_{\mathtt{a}}
M'_{\mathtt{a}}
M'_{\mathtt{b}}$,
$M'_{\mathtt{a}}
M'_{\mathtt{a}}
M'_{\mathtt{a}}
M'_{\mathtt{b}}$, and
$M'_{\mathtt{a}}
M'_{\mathtt{a}}
M'_{\mathtt{a}}
M'_{\mathtt{a}}$
and check that they
are of the form
\begin{displaymath}
  \begin{pmatrix}
1&    0 & 0 & 0 &0    \\
 c_i&    \hat C_i & 0 & 0 &0\\
    *&* & * & * & *\\
    *&* & * & * & *\\
    *&* & * & * & *\\
  \end{pmatrix}
\end{displaymath}
for $i=1,2,3,4,5$,
having the original matrices $M_i$ in their upper-left corner
and otherwise zeros in the first two rows of blocks.
 Thus, they simulate the original automaton on the states
$1,\ldots,d$:
It is easy to establish by induction that
multiplying the initial distribution vector ${\pi'}$ with a sequence of such matrices will produce a
vector $x'$ of the form
\begin{equation}
  \label{eq:result-form}
x' = (x_1\ \hat x \ 0\,\ 0\,\ 0)
\end{equation}
whose first $d$ entries
$x=(x_1\ \hat x)$ 
coincide with the entries of
the corresponding vector produced with the original start vector
$\pi$ and
the original
matrices~$M_i$.
If $x'$ is multiplied with $f'$, the result is the same as with
$x$ and the original vector~$f$.

One technicality remains to be discussed: Some  ``unfinished'' words in
$\{\mathtt{a},\mathtt{b}\}^*$  do not factor into codewords but
end in a partial
codeword \texttt{a}, \texttt{aa}, or \texttt{aaa}.
To analyze the corresponding matrix products, we look at the powers
$(M'_{\mathtt{a}})^i$, for $i=1,2,3$.
They have the form
\begin{equation}\label{eq:form-MaMa}
  \begin{pmatrix}
  1&  0 & 0 & 0 & 0\\
   0 & 0 & I & 0& 0\\
    * & * & * & *& *\\
    * & * & * & *& *\\
    * & * & * & *& *\\
  \end{pmatrix},\
  \begin{pmatrix}
  1&  0 & 0 & 0 & 0\\
   0 & 0 & 0& I & 0\\
    * & * & * & *& *\\
    * & * & * & *& *\\
    * & * & * & *& *\\
  \end{pmatrix},\
  \text{and }
  \begin{pmatrix}
  1&  0 & 0 & 0 & 0\\
   0 & 0 & 0& 0 & I\\
    * & * & * & *& *\\
    * & * & * & *& *\\
    * & * & * & *& *\\
  \end{pmatrix}.  
\end{equation}
and therefore, multiplying
 the vector $x$ from \eqref{eq:result-form} with them
yields
$ (x_1\ 0 \ \hat x\ 0\ 0)$,
$ (x_1\ 0\ 0 \ \hat x\  0)$, and
$ (x_1\ 0 \ 0 \ 0\ \hat x)$,
respectively.
If this is multiplied with $f'$, the result is the same as with
 the vector $x'$ in \eqref{eq:result-form}. 
Thus, as claimed after the statement of the lemma, {input sequences whose decoding process leaves
  a partial codeword $\mathtt{a}^i$ for $1\le i \le k-2$
  produce the same value
as if that partial codeword were omitted}.\footnote
 {Hirvensalo
  \cite
  {hirvensalo07} erroneously claimed that such
  incomplete inputs 
  give the value~0.
He mistakenly defined the vector $f'$ ($\mathbf{y}_3$ in his notation)
``analogously''
to the vector $\pi'$ ($\mathbf{x}_3$ in his notation),
thus padding it with zeros.
With this
definition, the result for incomplete inputs cannot be controlled at all.
}
  \end{proof}

\subsection{Proof of \autoref{thm-matrices-forward}\ref{2-matrices} }
\autoref{lem:reduce-to-2} blows up the number of states by a factor
that
depends on the number of matrices. Thus, it is advantageous to apply
it \emph{after} merging the start matrix into $\pi$, when the number of
matrices is reduced.

So we start with
the 5-pair instances of Neary~\cite{neary:PCP5:2015}
and construct five matrices
$A_i=A(v_i^R,w_i^R)$ from the reversed pairs.
We then
combine $A_5$ with the starting
distribution $\pi_1$ into $\pi_2 = \pi_1A_5$,
leaving $k=4$ matrices $A_1,A_2,A_3,A_4$ of
 dimension $d=6$. 

 We cannot apply the transformations of \hyperref[sec:step2]{Step~2},
because it changes the first row, and state~1 would no longer be
absorbing, which precludes the application of \autoref{lem:reduce-to-2}.
 
Thus, we apply
\autoref{lem:reduce-to-2} to
the matrices $A_1,A_2,A_3,A_4$,
resulting in two matrices
$M_{\mathtt{a}}'$ and
 $M_{\mathtt{b}}'$
 of dimension
 $(k-1)(d-1)+1=16$.
 The new start vector $\pi_3$ is $\pi_2$ padded with zeros,
 and the end vector $f$ is still the first unit vector (of dimension 16).

We would now like to use the transformation of \hyperref[sec:step2]{Step~2} to make
the matrices column-stochastic.
However, since we have replaced the initial
start vector $\pi_1$ by $\pi_2$,
the entries of the start vector
after the transformation 
would not longer sum to~0, a property that is
crucial for eventually making the matrices stochastic in \hyperref[sec:step4]{Step~4}.

Therefore we have to achieve column sums 1
in a different way, with
an adaptation of the method of \hyperref[sec:step3]
{Step~3} (\autoref{sec:row-1}),
adding an extra state;
 
\subsection{Step 2\texorpdfstring{$'$}{'}: Making column sums 1 with
  an extra state}
\label{sec:step2'}
We add an extra row $s_i$
and an extra column to each matrix, ensuring that all
column sums become 1.
\begin{equation}
  B_1=
  \begin{pmatrix}
    M_{\mathtt{a}} & 0\\
    s_1
    & 1\\
  \end{pmatrix},
  \
  B_{2} =
  \begin{pmatrix}
    M_{\mathtt{b}} & 0\\
    s_2
    & 1\\
  \end{pmatrix}
\end{equation}
We now have two $17\times 17$ matrices $B_1,B_2$ with column sums
1. 
The remaining steps (\hyperref[sec:step3]{Steps~3}--\hyperref[sec:step5]5) are as before, with the appropriate
adaptations, and they add another row and column. (We may have to choose a smaller constant $\alpha$ in
\hyperref[sec:step4]{Step~4}.)
This concludes the
proof of \autoref{thm-matrices-forward}\ref{2-matrices}. \qed

We mention that the combined effect of
\hyperref[sec:step2']{Steps~2$'$}--\hyperref[sec:step5]5 can also be achieved by applying
the sharpened version of Turakainen's 
Theorem~\cite[Theorem~1(i)] {Turakainen_1975} from 1975,
which converts any generalized PFA into a 
PFA with just 2 more states.
(This method could already have been used by
Blondel and Canterini
in 2003 and by Hirvensalo in 2007, who used 4 extra states to achieve the same effect, see
  \cite[Steps 3 and~4]{blondel-canterini-03}
and  
\cite[Steps 4 and~5]{hirvensalo07}.) 
Since we had almost all tools available, we have chosen to
describe the conversion directly.
In fact, our procedure more or less parallels Turakainen's proof,
except that we do not have to treat the vector $f$ because it is
already a 0-1-vector
(Step~1 of~\cite{Turakainen_1975}).
\hyperref[sec:step2']{Steps~2$'$} and~\hyperref[sec:step3]3 are usually treated together as one step.
Our \hyperref[sec:step4]{Step~4}, the conversion to positive matrices, is
 more streamlined than in other proofs in the
literature.\footnote{We take this occasion to point
  out a slight mistake in the statement and proof of 
\cite[Theorem~1(i)]{Turakainen_1975}.
The cutpoint is not $1/(n+2)$ but it must be modified by
the quantity~$d$ from 
\cite[top of p.~30]{Turakainen_1975}.
For the case $\pi f>0$, the original statement
of \cite[Theorem~1(i)]{Turakainen_1975}
can still be recovered from
the stronger statement of
\cite[Theorem~1(iv)]{Turakainen_1975} by rescaling.}

\subsection{Proof of
  \autoref{thm-matrices-forward}\ref{2-matrices-fixed}  (two fixed matrices)}

This is obtained by adapting the Proof of
\autoref{thm-matrices-forward}\ref{2-matrices}
in the same way as for
\autoref{thm-matrices-forward}\ref{5-fixed}
(\autoref{sec:Mati}).
Instead of Neary's 5-pair instances,
we take the instance
$(v_1,w_1),\,\ldots, \,(v_7,w_7)$
of
 Matiyasevich and S{\'e}nizergues 
and construct seven matrices
$A_i=A(v_i,w_i)$ from these pairs (without reversing the words
$v_i$ and~$w_i$).
We then
combine $A_1$, which must be the first matrix in the product, with the starting
distribution $\pi_1$ into $\pi_2 = \pi_1A_1$.
The remaining matrices $A_2,A_3,\ldots,A_7$ are fixed. Only $\pi_2$
is an input to the problem.

The remainder of the proof is the same as in \autoref{thm-matrices-forward}\ref{2-matrices}.
 We apply
\autoref{lem:reduce-to-2} to
$k=6$ matrices $A_i$ of dimension $d=6$, resulting in two fixed matrices
$M_{\mathtt{a}}'$ and
 $M_{\mathtt{b}}'$
 of dimension
 $(k-1)(d-1)+1=26$.
The conversion to a PFA requires two more states, and this leads to
\autoref{thm-matrices-forward}\ref{2-matrices-fixed}.\footnote
{We mention that \emph{one} of the two matrices in 
  \autoref{thm-matrices-forward}\ref{2-matrices} can also be held fixed:
Neary's PCP contains one fixed word pair
$(v_4,w_4)=(\mathtt{1},\mathtt{0})$
\cite[Theorem~11]{neary:PCP5:2015}.
We see that the matrix $M'_{\mathtt{a}}$ in the construction of
\autoref{lem:reduce-to-2} contains only one of the matrices
$M_i$ of the  original $k$-state automaton. (In the example shown, it is
the matrix~$M_4$.) If we arrange that this is the matrix
constructed from the fixed pair
$(v_4,w_4)$, then  $M'_{\mathtt{a}}$, and hence the corresponding
stochastic $18\times 18$
matrix in the final PFA, will be fixed.}
\qed

 \section{Proof of \autoref{thm-matrices-backward} 
   (variable end 
   vector \texorpdfstring{${f}$}{f})}

One can relax the requirement that $f$ is a 0-1-vector and
allow arbitrary values.
If the values are in the interval $[0,1]$
we can think of the entry $\eeta_i$ as a probability in a final acceptance
decision, if the PFA is in state $i$
after the input has been read.
Another possibility is that
 $\eeta_i$ represents
 a \emph{prize} or \emph{reward} that
that is gained when the process stops
in state~$i$.
Then
 $\eeta_i$
does not need to be restricted to the interval $[0,1]$.
In this view, instead of
the acceptance probability, we compute the \emph{expected} reward
of the automaton, as in game theory.
We call $\eeta_q$
the \emph{output values},
and $\eeta$ and
the \emph{output vector} or
the \emph{end vector}, in analogy to the start vector $\pi$.

Since the transition matrices $B_i$ after
\stepref 2 are column-stochastic and
positive by \autoref{lem:column-stochastic}, we transpose
them to produce stochastic matrices, which can be directly used as transition
matrices.
This will reverse the order of the matrices in the product
and swap the start vector with the end vector:
\begin{align}\nonumber
  \pi_2B_{i_1}B_{i_2} \ldots B_{i_m}f_2
&=
    f^T_2B^T_{i_m}B_{i_{m-1}}^T \ldots B_{i_1}^T\pi_2^T 
\\&=
  \pi_6B^T_{i_m}B_{i_{m-1}}^T \ldots B_{i_1}^Tf_6
  \label{eq:reversed-product}
\end{align}
We may prefer to have the end vector $f$ in the interval $[0,1]$.
Thus,
we replace $f_6=\pi_2^T $
by $f_7 = \bigl( (2,2,2,2,2,2)+\pi_2) \bigr)^T/12$,
in analogy to
\autoref{sec:step5}.
This vector is positive and has sum~1,
and in particular, the entries lie between 0 and 1.
The effect is to increase
the value of \eqref{eq:reversed-product} by 1/6.

 We take
 the 5-pair instances of Neary~\cite{neary:PCP5:2015}
(\autoref {sec:select-PCP})
and construct five matrices
$A_i=A(v_i^R,w_i^R)$ from the reversed pairs.
The words $v_i$ and $w_i$ in these instances are nonempty,
as required by
\autoref{lem:column-stochastic}
\cite[Proof of Theorem~11]{neary:PCP5:2015}.
We convert them to column-stochastic matrices $B_i$
and use the transposed matrices $B_i^T$. We know
that
$(v_5^R,w_5^R)$ must be used at the beginning of the
PCP solution $i_1i_2\ldots i_m$ ($i_1=5$)
and nowhere else.
Hence
$B_5^T$
must be used at the end of the product in
\eqref
{eq:reversed-product},
and we can merge it with $f_6$
into an end vector $f_8 = B_5^Tf_7 = B_5^T\pi_2^T$.
The four remaining matrices
$B_1^T,B_2^T,B_3^T,B_4^T$ form the set $\mathcal{M}$.
The starting distribution $\pi_6 = f_2^T$ is a unit vector,
i.e., there is a single deterministic start state.
This proves 
\autoref{thm-matrices-backward}\ref{backward}
for $\lambda=1/6$.

\subsection{Changing the cutpoint
\texorpdfstring{$\lambda$}{λ}
  by manipulating the end vector
 \texorpdfstring{${f}$}{f}}
\label{sec:manipulate}
When the end vector $f$ of a PFA is under control, one can change the
cutpoint $\lambda$ to any desired value
in the open interval between $0$ and $1$:
Adding a constant $K$ to all output values $f_i$ will increase
the expected output value by $K$, and scaling by a positive constant
will affect the expected output value in the same way.

Thus, by changing all $f_i$ to $\alpha f_i$ for $0<\alpha<1$, one may
decrease $\lambda$ to any positive value.
 By changing all $f_i$ to $1-\alpha +\alpha f_i$ for $0<\alpha<1$, one may
 increase $\lambda$ to any value less than~1.
If $f$ is not constrained to the interval $[0,1]$, one can reach any
real value~$\lambda$.

 \subsection{Proof of
   \autoref{thm-matrices-backward}\ref{backward-fixed} (fixed
  transition matrices)}
\label{sec:back-fixed}


We would like to apply the approach of
\autoref{thm-matrices-backward}\ref{backward}
to the the seven word pairs
 $(v_1,w_1),\,\allowbreak\ldots, 
 \,(v_7,w_7)$
from
\autoref{sec:Mati}.
They should fulfill the assumption of
\autoref{lem:column-stochastic} that none of the words is empty.
However, one of the rules of the
3-rule semi-Thue system of
Matiyasevich and S{\'e}nizergues,
the rule $x\bar x\longrightarrow \epsilon$,
contains
an empty word,
see \cite[System~$S_5$, (23)]{matiyasevich2005decision}.
 The reduction of 
Claus
\cite{\clausEATCS}\ 
(see also
\cite[Theorem~4]{halava07-Claus-instances})
would translate this into a PCP pair with an empty word.

Luckily, this reduction can be patched
to yield seven pairs of nonempty words $v_i$ and $w_i$.
We refer to
\autoref{sec:semi-Thue} for details.

 As above, we form from these pairs the stochastic matrices
$B_1^T,\ldots,B_7^T$.
Only $B_1^T$ is a variable matrix, and all other matrices are fixed.
In the product
\begin{math}
 \pi_6B^T_{i_m}B_{i_{m-1}}^T \ldots B_{i_1}^Tf_7
\end{math},
we merge the first matrix with $\pi_6$ into
$\pi_8 := \pi_6B^T_{i_m} = \pi_6B^T_2$, which becomes the fixed start
vector,
and the last matrix 
into
$f_8 := B_{i_1}^Tf_7 = B_{1}^Tf_7$, which becomes the only input to the
problem.
The remaining five matrices
$B_3^T,B_4^T,B_5^T,B_6^T,B_7^T
$ form the fixed set $\mathcal{M}$.
This proves 
\autoref{thm-matrices-backward}\ref{backward-fixed}
for $\lambda=1/6$, and as above, it can be
changed to any cutpoint $\lambda$.
\qed

\section{Outlook}
\label{sec:outlook}
The natural question is to ask for the smallest number of states for which
PFA Emptiness is undecidable.
Even if we consider generalized PFAs (rational-weighted automata),
we could not reduce this number below~6.
Claus~\cite[Theorem 7 and Corollary, p.~155]{claus81} showed that
the emptiness question can be decided for
PFAs with two states.

A matrix dimension of six seems to be the minimum required in order to
carry enough information
to model concatenation of word
pairs and allow testing for equality by a quadratic expression like~\eqref{eq:almost-zero},
even in weighted
automata. Undecidability for five or perhaps even three states would
require some completely different (perhaps geometric) approach.

It would be an interesting exercise to trace back the
undecidability proof of
Matiyasevich and S{\'e}nizergues
\cite{matiyasevich2005decision}
to its origins
and explicitly work out the fixed matrices of
\autoref{thm-matrices-forward}\ref{5-fixed} or
\ref{thm-matrices-backward}\ref{backward-fixed}.
\ifanonymous
\else
For one of the weaker results mentioned after
\autoref{thm-matrices-backward},
\cite[Theorem~4b]{r-pfaeu-24},
one particular matrix from a set of 52 fixed $11\times 11$ matrices is
shown in 
\cite[Section~7.3]{r-pfaeu-24}. 

\fi

PFAs can be have other merits than just a small number of states
and input symbols. We discuss some of these criteria.

\paragraph{Positive and doubly-stochastic transition matrices.}
In our results,
the
 transition matrices can always be assumed to be
 strictly positive and 
 sometimes also
 doubly-stochastic.  Whenever this is the case,
 we have mentioned it in
 Theorems~\ref{thm-matrices-forward} and~\ref{thm-matrices-backward}.

\subsection{Obtaining a substantial gap between acceptance and rejection}
\label{sec:amplification}

As seen in formula
\eqref{eq:almost-zero}, the acceptance probability barely rises
above the threshold $0$ when the input represents a solution of the
PCP.
(This tiny gap is further reduced by multiplying it with
 $\alpha^m$ in
 \hyperref[sec:step4]{Step~4}.)
Thus, our constructions 
depend on the capability of a PFA to ``detect'' minute fluctuations
of the acceptance probability above the cutpoint.
This statement applies to all undecidability proofs in the Nasu--Honda
line of descent.

By contrast,
the proof of 
Condon and Lipton
\cite{condon-lipton-1989} from 1989
gives a more robust result\ifanonymous\else, see also 
\cite[Section~4]{r-pfaeu-24}\fi
:
 For any $\eps>0$, it yields a PFA
such that the largest acceptance probability is either $\ge 1-\eps$ or
$\le \eps$, and
the problem to detect which of the two cases holds is undecidable.
 Undecidability is derived from the halting problem for 2-counter
machines,
and the number of states of the PFA
is beyond control.

Luckily, our results can be strengthened to
 even surpass the strong separation achieved in the
Condon--Lipton construction, by the following
gap amplification technique of
Gimbert and Oualhadj, which
we formulate in slightly generalized form.

\begin{restatable}
[Gimbert and Oualhadj 2010
\cite{gimbert-oualhadj-2010:PFA}]
  {theorem}
  {GapAmplification}
  \label{thm-amplify}
  We are given a PFA $A$ with input alphabet $\Sigma$ and $d$ states,
  with an output vector $f\in [0,1]^d$.
  We denote its
acceptance probability for an input $u\in \Sigma^*$ 
by $\phi(u)$.
Let $\lambda_A,\lambda_B$ be two thresholds with
$0\le\lambda_A \le1$ and
$0<\lambda_B <1$.

Then, from the description of $A$, we can construct a PFA $B$ with input alphabet
$\Sigma\cup\{\mathtt{end},\mathtt{check}\}$
with $2d+3$ states, a single
start state and a single accepting state,
with the following
property.
\begin{enumerate}[1.]
\item
    If every input $u\in \Sigma^*$ for $A$ has acceptance probability $\phi(u)\le\lambda_A$,
then every input for $B$ is accepted by $B$ with probability $\le\lambda_B$.
\item If $A$ has an input $u\in \Sigma^*$ with acceptance probability $\phi(u)>\lambda_A$,
  then $B$ has inputs with acceptance probability arbitrarily close to~$1$.
\end{enumerate}
\end{restatable}

\medskip

This was proved
for $\lambda_A=\lambda_B=1/2$ by
Gimbert and Oualhadj~\cite{gimbert-oualhadj-2010:PFA}
in 2010
in order to show that
it is undecidable
whether the conclusion of Case~2 
for a PFA $B$ holds (the ``Value~1 Problem'').
The construction and proof was simplified by
Fijalkow~\cite{Fijalkow-2017-SIGLOG}.
The generalization to arbitrary
$\lambda_A$ and $\lambda_B$ is not difficult,
and in addition,
we have included the precise statement about the number $2d+3$ of states.
For completeness, we present a proof in
\autoref{sec:amplification-a}.
It essentially follows
Fijalkow's construction, and
it eliminates an oversight
in
\cite{gimbert-oualhadj-2010:PFA}
and~\cite{Fijalkow-2017-SIGLOG}.

With this technique, one can achieve an arbitrarily large gap
with a moderate increase in states, roughly by a factor of 2.
In particular, from
\autoref{thm-matrices-backward}\ref{backward}, we get a PFA $B$
with 6 matrices of size $15\times 15$, which exhibits the strong
dichotomy expressed in
\autoref{thm-amplify},
for any $\lambda_B>0$.

This construction does not preserve the property
of being a PFA with fixed transition matrices. In  the
PFA $B$ constructed
in \autoref{sec:amplification-a}
(see \autoref{tab:expanding}), the transitions depend both
on the starting distribution $\pi$ and on the final vector $\eeta$ of~$A$.

\subsection{Uniqueness}

In \autoref{thm-matrices-forward}\ref{4-positive} and
\autoref{thm-matrices-backward}\ref{backward},
the constructed PFA has the property that the
recognized language contains at most one word.
As with the large gap guarantee in \autoref
{sec:amplification}, this leads to a stronger statement where
 the problem gets the nature of a \emph{promise problem}.

 Neary \cite{neary:PCP5:2015} derived his undecidable PCP
 instances from \emph{binary tag systems}.
 A binary tag system performs a deterministic
 computation.
 It follows from the
 correctness argument of the simulation that the PCP solution is
 unique if it exists, apart from the obvious possibility of repeating
 the solution arbitrarily.
In \autoref{sec:select-PCP},
 we have excluded the last possibility by
 fixing the end pair $(u_5^R,w_5^R)$ to be the first pair and
 removing it from the list. Thus, uniqueness holds for the PFA
 in
Theorems \ref{thm-matrices-forward}\ref{4-positive} and~\ref{thm-matrices-backward}\ref{backward}.

However,
 in the conversion to
 a binary input alphabet
 for Theorem~\ref{thm-matrices-forward}\ref{2-matrices}
 (\autoref{sec:reduce-to-2}),
 uniqueness is lost:
 We have seen that we may always add
 \texttt{a} or \texttt{aa} 
 to a solution.
 We don't see an easy way to eliminate these extra solutions
 without increasing the number of states.


For 
Theorems~\ref{thm-matrices-forward}\ref{5-fixed} (\autoref{sec:Mati})
and~\ref{thm-matrices-backward}\ref{backward-fixed}
(\autoref{sec:back-fixed}),  we used a PCP that models
semi-Thue systems.
However, even some
3~rule semi-Thue systems, like the one
of~\cite{matiyasevich2005decision}, would have 
at most one derivation (which we have not checked),
uniqueness would not survive the reduction to the PCP,
see the remark in \autoref{sec:semi-Thue}, \refToStepA1, after \autoref
{prop:PCP-properties}.
Thus we cannot claim uniqueness in these cases.

\subsection{Simple probabilistic automata}
In a \emph{simple} PFA, all probabilities are $0$, $\frac12$, or~1
  \cite[Definition~2]{gimbert-oualhadj-2010:PFA}.
We have constructed our PFA with decimal fractions, but it would
not be hard to switch to binary fractions.
Before 
\hyperref[sec:step4]{Step~4}, the number of states should be increased
to a power of two, so that the entries of $J$ become binary
fractions as well.
Once all transition probabilities are binary fractions, the PFA can be
converted to a simple PFA by
padding each input symbol with sufficiently many dummy symbols so that
the PFA has time to make its random decisions with a sequence of
unbiased coin flips.
\ifanonymous\else
see \cite[Section~4.4, proof of Theorem~1, item~(b)]{r-pfaeu-24}
\fi
Thus the results can be achieved with simple PFAs, but with
a larger 
 number of matrices and
a larger 
 number of states.
The precise quantities would depend on the lengths of the words $v_i$
and $w_i$ in the PCP.

\phantomsection
\addcontentsline {toc}{section}{\numberline {--}References}

\bibliography{PFA-undecidable}

\appendix

\section{Proof of \autoref{thm-amplify} (gap amplification)}
\label{sec:amplification-a}


\GapAmplification*

\begin{proof}
We first describe $B$ and explain its operation under the
assumptions that $\lambda_A=
1/2$ and $f$~is a 0-1-vector;
in other words, $A$ has a set $F$ of accepting states and a
complementary set $R$ of rejecting states.

The input for the PFA $B$
is structured into \emph{rounds}. Each round is
of the form
\begin{equation}\label{eq:round}
\ldots    \mathtt{check}\
u_1\ \mathtt{end}\ u_2\ \mathtt{end}\
\ldots\ \mathtt{end}\ u_n\  \mathtt{end}\ u_{n+1}\ 
  \mathtt{check} \ldots
\end{equation}
with $u_i\in \Sigma ^*$.
The \texttt{check} symbols delimit each round
and
separate it from the adjacent rounds.
The \texttt{end} symbols
separate successive strings $u_i$, which are taken as inputs for
(a simulation of)~$A$.
The last string $u_{n+1}$ is ignored: It has no effect on the outcome.

\begin{figure}[th]
  \centering
  \includegraphics[scale=0.96]{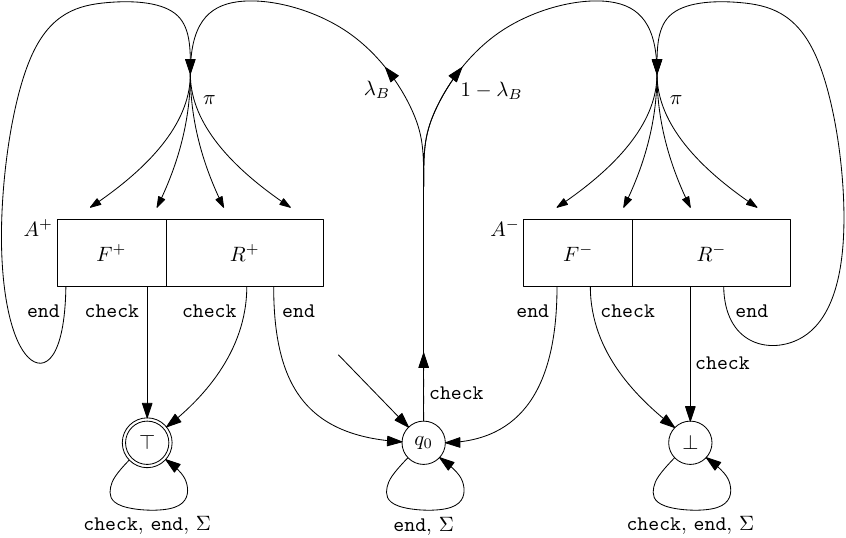}
  \caption{The expanding PFA $B$}
  \label{fig:amplify-PFA}
\end{figure}

The PFA $B$ contains two copies $A^+$ and $A^-$ of the machine $A$.
 There are three  additional states: a start state $q_0$, and two
 absorbing states $\top$ and $\bot$.
The illustration in Figure~\ref{fig:amplify-PFA}
 shows the two copies of $A$ split into
accepting states $F$ and rejecting states $R$.
Each bundle of arrows leading into $A^+$ and $A^-$ indicates a transition
according to the input distribution $\pi$ of $A$. 

The behavior of the machine can be described as follows:
At the beginning of each round, the machine is
either in one of the two absorbing states or
in state $q_0$.
If it is
in state $q_0$, it flips a biased coin, and with probabilities $\lambda_B$ and $1-\lambda_B$, it does one of the
following two things:

\begin{enumerate}
\item [($+$)]
  If
  \emph{all} strings 
  $u_1,\ldots,u_n$ of this round were accepted by $A$,
  the machine $B$ {accepts} the input, moving to the absorbing
  state $\top$.
  Otherwise, it proceeds to the next round.
  \smallskip
\item [($-$)]
  If
  \emph{all} strings 
  $u_1,\ldots,u_n$ of this round were rejected by $A$,
the machine $B$ {rejects} the input,
moving to the absorbing 
state $\bot$.
  Otherwise, it proceeds to the next round.
\end{enumerate}

We call the three possible outcomes of a round
ACCEPT, REJECT, and INDECISION.
If the PFA never takes a decision before the input is exhausted,
it will end in some nonabsorbing state, and hence it will
also reject the input.


It is easy to write down the probabilities of the outcomes of a
round~\eqref{eq:round}:
\begin{align}
  \label{eq:plus}
  \Pr[\text{ACCEPT}]
 & = \lambda_B
 \phi(u_1)\phi(u_2) \cdots \phi(u_n) \cdot 
 \\
  \Pr[\text{REJECT}]
 & =(1-\lambda_B)
 \label{eq:minus}
 \bigl(1-\phi(u_1)\bigr)\bigl(1-\phi(u_2)\bigr) \cdots \bigl(1-\phi(u_n)\bigr) 
\end{align}

Case 1: If
$\phi(u_i)\le1/2$ for 
every $u_i$, then $
\phi(u_i)\le 1-
\phi(u_i)$, and
it follows that
\begin{displaymath}
      \frac{  \Pr[\text{ACCEPT}] }
      { \Pr[\text{REJECT}]}
      \le
\frac{\lambda_B}{1-\lambda_B},
\end{displaymath}
and hence
\begin{displaymath}
      \frac{  \Pr[\text{ACCEPT}] }
      {  \Pr[\text{decide}]}
=      \frac{  \Pr[\text{ACCEPT}] }
      {  \Pr[\text{ACCEPT}] +  \Pr[\text{REJECT}]}
      \le \lambda_B,
\end{displaymath}
For the purpose of the overall analysis, it helps
to imagine that the machine 
\emph{first} makes up its mind \emph{whether} a decision is made, and only
then decides \emph{which} of the branches ACCEPT or REJECT is taken,
as represented in Figure~\ref{fig:amplify}.

\begin{figure}[htb]
  \centering
\begin{tikzpicture}[
roundnode/.style={circle, draw=black, minimum size=7mm},
squarednode/.style={rectangle, draw=red!60, fill=red!5, very thick, minimum size=5mm},
boxnode/.style={rectangle, draw=black, minimum size=5mm},
]

\node[boxnode,anchor=west] at (1,2.4)     (undec1)             {INDECISION};
\node[boxnode,anchor=west] at (5.5,4.9)     (eq1)             {ACCEPT};
\node[boxnode,anchor=west] at (5.5,3.5)     (diff1)             {REJECT};
\node[boxnode,anchor=west] at (1,4.2)        (dec1)       {decide};
\node[roundnode,label={left:Case 1: $\phi(u_i) \le \frac12\quad $}
] at (-1.2,3.3)        (start1)     {}  ;

\draw[->] (start1) -- (dec1.west) node[above,midway,sloped]{(rarely)};
\draw[->] (start1) -- (undec1.west) node[below,midway,sloped]{(often)};
\draw[->] (dec1.east) -- (eq1.west) node[midway,above,yshift=0.5mm]
{
  $\le \lambda_B$
};
\draw[->] (dec1.east) -- (diff1.west) node[midway,below,yshift=-1mm]
{
  $\ge 1-\lambda_B$
};

\node[boxnode,anchor=west] at (1,-1.1)     (undec2)             {INDECISION};
\node[boxnode,anchor=west] at (5.5,1.4)     (eq2)             {ACCEPT};
\node[boxnode,anchor=west] at (5.5,0)     (diff2)             {REJECT};
\node[boxnode,anchor=west] at (1,0.7)        (dec2)       {decide};
\node[roundnode,label={left:Case 2: $\phi(u) > \frac12\quad $}
] at (-1.2,-0.2)        (start2)     {}  ;

\draw[->] (start2) -- (dec2.west) node[above,midway,sloped]{(rarely)};
\draw[->] (start2) -- (undec2.west) node[below,midway,sloped]{(often)};
\draw[->] (dec2.east) -- (eq2.west) node[midway,above]{$\ge1-\eps$\ \quad \null};
\draw[->] (dec2.east) -- (diff2.west) node[midway,below,yshift=-0.5mm]{$\le\eps$};


\end{tikzpicture}

\caption{The outcome of processing one round by the PFA $B$.
Case 2 assumes a large number $n$ of repetitions of~$u$.}
  \label{fig:amplify}
\end{figure}
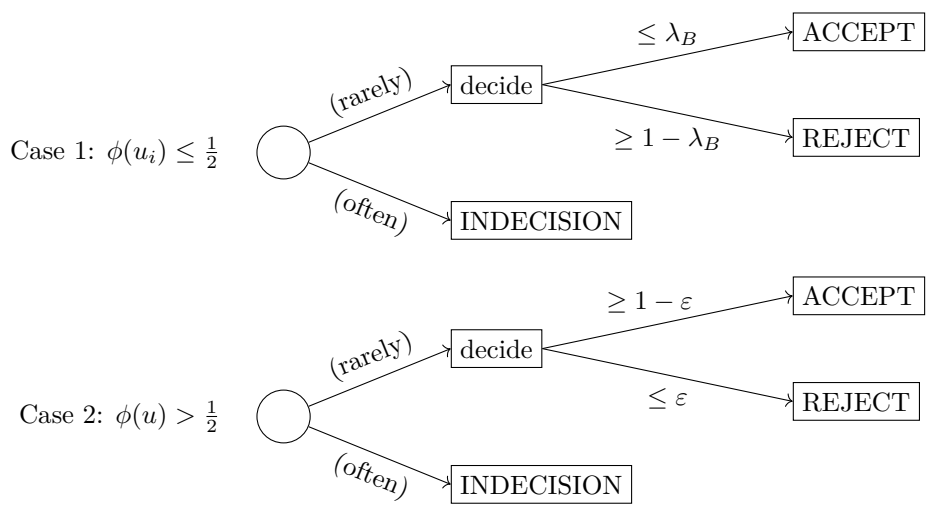

Then, for an input that consists of any number of rounds, whenever the
decision is taken, $B$ favors REJECT over 
ACCEPT with a ratio at least $(1-\lambda_B) : \lambda_B$.
To this, we must add the probability of rejection
due to continuing INDECISION, which is also in favor of rejection.
This proves statement~1 of the theorem.

Case 2:
We take a word $u$ with $\phi(u)>1/2$ and repeat it a large number
of times. In this way,
for a round of the form $\ldots\mathtt{check}\
(u\ \mathtt{end})^n \mathtt{check}\ldots$,
we can make the ratio
between ACCEPT \eqref{eq:plus}
and REJECT \eqref{eq:minus}
as large as we like:
\begin{equation}
\label  {eq:eps}
      \frac{  \Pr[\text{ACCEPT}] }
      {  \Pr[\text{decide}]}
=      \frac{  \Pr[\text{ACCEPT}] }
      {  \Pr[\text{ACCEPT}] +  \Pr[\text{REJECT}]}
      \ge 1-\eps,
    \end{equation}
    for any $\eps>0$.
    The probability of INDECISION will also approach~1, but this can
    be compensated by a large number of rounds:
We create an input of the following form:
\begin{displaymath}
\bigl(    \mathtt{check}\
(u\ \mathtt{end})^n \bigr)^m
\end{displaymath}
We first choose $n$ large enough to satisfy~\eqref{eq:eps} in each round.
Then, no matter how tiny the probability of a decision in a single round is,
the probability of continued INDECISION can be made arbitrarily small
by increasing $m$.
The input is rejected if either
(a) no decision is taken or (b) the PFA takes a decision and it
decides to REJECT.
Both probabilities can be made arbitrarily small, and
this proves statement~2 of the theorem.

\begin{table}[bthp]
  \centering
  \begin{tabular}{|l|l|l|l|}\hline
    state of $B$&input symbol&transition to&with probability\\
    \hline
    $q_0$ &\texttt{check} & $q_j^+$ & $\lambda_B\pi_j$\\
    & & $q_j^-$ & $(1-\lambda_B)\pi_j$\\
    \cline{2-4}
                &\texttt{end}, $\Sigma$ & $q_0$ & 1\\
    \hline
    $q_i^+$ &\texttt{end} & $q_j^+$ & $f_i\pi_j$\\
    &&$q_0$ & $1-f_i$\\
    \cline{2-4}
                &\texttt{check} & $\top$ & 1\\
    \cline{2-4}
 &$\Sigma$ &  $q_j^+$ & as from $q_i$ to $q_j$
                                   in $A$\\
    \hline
    $q_i^-$ &\texttt{end} &$q_0$ & $f_i$\\
    && $q_j^-$ & $(1-f_i)\pi_j$\\
    \cline{2-4}
                &\texttt{check} &$\bot$ & $1$\\
    \cline{2-4}
                &$\Sigma$ &  $q_j^-$ & as from $q_i$ to $q_j$ in $A$\\
    \hline
  \end{tabular}
  \caption{Transitions in $B$ for all states except the absorbing
    states $\top$ and $\bot$} 
  \label{tab:expanding}
\end{table}

We still have to discuss the cases that
 $f$ is not a 0-1-vector or
 $\lambda_A\ne
 1/2$.
 Consider a state~$i$ with $0< f_i<1$.
 If $B$ is in such a state and the \texttt{end} symbol is read, then
with probability $f_i$, it behaves like an accepting state,
and with probability $1-f_i$, it behaves like a rejecting state.
This is reflected the transition table of \autoref{tab:expanding}.

Finally, if
 $\lambda_A\ne1/2$, we have seen in
 \autoref{sec:manipulate} how we can modify $f$ to get an equivalent
 PFA with $\lambda_A=1/2$.
Thus, the theorem is proved in full generality.
\end{proof}




\section[From
  a semi-Thue system with \texorpdfstring{$k$}{k} rules to a PCP
  with \texorpdfstring{$k+4$}{k+4} nonempty pairs
]{
  From 
  a semi-Thue system with \texorpdfstring{$k$}{k} rules to a PCP
  with \texorpdfstring{$k+4$}{k+4} pairs of nonempty words}
\label{sec:semi-Thue}

\newcommand\leftbound{
  [}  
\newcommand\rightbound{
 ]}
\newcommand\middleseparator{*} 

The construction is due
to Claus \cite{\clausEATCS}
from 1980,
and it was refined by
Halava, Harju, and Hirvensalo in 2007
  \cite[Theorem~6]{halava07-Claus-instances},
 with a condensed and slightly different proof.
We give a variation of the proof, extending it to ensure that none of the
words in the pairs is {empty}.
 
The following version of the PCP  is suitable for our needs
\cite[Section 3.2]{Harju1997}: 
%
\begin{quote}
\emph{The Generalized Post Correspondence Problem (GPCP)}.\\
Given a sequence of pairs of words $(v_1,w_1), (v_2,w_2), \ldots
 (v_k,w_k)$,
is  there a 
 sequence
 $i_2,\ldots, i_{m-1}$
of indices $i_j \in \{3,\ldots,k\}$ such that
\begin{displaymath}
  v_1v_{i_2}v_{i_{3}}\ldots v_{i_{m-1}}v_2
=   w_1w_{i_2}w_{i_{3}}\ldots w_{i_{m-1}}w_2 \ ?
\end{displaymath}
\end{quote} 
We have already encountered
such special roles
of the start pair $(v_1,w_1)$ and the end pair $(v_2,w_2)$,
both in connection with
the 5-pair PCPs of Neary
and the 7-pair PCPs of
Matiyasevich and S{\'e}nizergues.
However, there it was
 a property of the PCP \emph{instances}, of which this restricted
 form of the solutions was a \emph{consequence}.
In case of the GPCP, this is a \emph{requirement} on the solution; it is part of the
``rules of the game.''



There is also the
\emph{Modified Post Correspondence Problem} (MPCP), where only the start pair $(v_1,w_1)$ is
prescribed. It
is often used as an
intermediate problem when reducing the Halting Problem for Turing
machines to the PCP.
This
MPCP (or the GPCP) can be reduced to the classic PCP with some extra effort, by
modifying the word pairs,
see for example
\cite[Lemma 8.5
]{hopcroft79} or
\cite[p.~189]{Sipser}.
We describe this procedure, which follows the standard pattern,
in \refToStepA4 of the proof below.

In our situation,
the GPCP is actually the more convenient version of the
problem.
Although we don't need the PCP version of the theorem, 
we include it for completeness, and
also since the proof in \cite[Theorem~6]{halava07-Claus-instances}
contains a slight error, which makes is 
invalid
in the very context of empty words.







The theorem starts from a \emph{semi-Thue system}. 
A {semi-Thue system} 
is specified by
a set of \emph{replacement rules} of the form
  $l\longrightarrow r$, where $l$ and $r$ are words over some
  alphabet. Such a rule means that a word $u$ can be transformed
  by replacing some occurrence of 
  $l$ in $u$ by $r$.
This is called a one-step derivation, and it is written as $u\to v$,
meaning that $u=xly$ and $v=xry$ for some words $x,y$.
As usual,  $\xrightarrow*$ denotes the transitive and reflexive closure
of this relation.

\begin{theorem}
  \label{lem:claus-instance}
  Suppose there is a semi-Thue system 
  with $k$ rules
  $l_i\longrightarrow r_i$ over some alphabet~$\Sigma$, and a word $u_\infty\in\Sigma^*$,
  for which the following problem 
  \textup(the \emph{individual word problem} or
\emph{individual accessibility problem} with fixed target $u_\infty$\textup)
  is undecidable:
  \begin{quote}
    Given a word $u_0\in \Sigma^*$, is
    $u_0 \xrightarrow
    {*} u_\infty$?    
  \end{quote}
  Then there is a list of $2(k+4)-1$ nonempty words
 $v_1,v_2,v_3,\ldots, v_{k+4}$ and $w_2,w_3, \ldots,w_{k+4}$
 over the binary alphabet $\{\mathtt{a},\mathtt{b}\}$ for which
 the following problem is undecidable:
\begin{quote}
  Given a nonempty word $w_1 \in\{\mathtt{a},\mathtt{b}\}^*$,
does the
  GPCP with the $k+4$ word pairs
   $(v_1,w_1),\allowbreak (v_2,w_2),\allowbreak \ldots,\allowbreak
 (v_{k+4},w_{k+4})$ have a solution?
\end{quote}
The construction can be modified so that the same question with the
\textup(ordinary\textup) PCP is undecidable.
\end{theorem}

\begin{proof}
  We transform the semi-Thue system into a PCP in several steps.
  
\paragraph{Step 1: Representing a derivation as a GPCP solution.}
\label{para:step1}

We start with a simpler procedure that creates $k+3+|\Sigma|$ word pairs.
We expand the alphabet to the larger alphabet $\Sigma' :=
\Sigma\cup
\{\texttt{\textup\#}\}$
with a new separator symbol \texttt{\#}.
We use the following
word pairs:
\begin{itemize}
\item the \emph{start pair} $(v_1,w_1)= ( \texttt{\#}
, \texttt{\#}
u_0\texttt{\#})$.
\item the \emph{end pair} $
(v_2,w_2)= 
  (  \texttt{\#}u_\infty\texttt{\#},
   \texttt{\#})$.
  \end{itemize}
In addition, we have the following $k+|\Sigma|+1$ word pairs:
  \begin{itemize}
\item for every rule $l\longrightarrow r$, the \emph{rule pair} $( l,r)$
\item for every symbol $\sigma \in \Sigma' 
  $,
the \emph{copying pair} $(\sigma,\sigma)$
\end{itemize}

The idea is that a derivation
\begin{equation}
  \label{eq:derivation}
u_0\to u_1 \to \dots \to u_{n-1} \to
u_n=u_\infty  
\end{equation}
is represented by the PCP solution
\begin{displaymath}
  v_1v_{i_2}v_{i_{3}}\ldots v_{i_{m-1}}v_2
=   w_1w_{i_2}w_{i_{3}}\ldots w_{i_{m-1}}w_2 
=\texttt{\#} 
u_0\texttt{\#} u_1 \texttt{\#} \ldots \texttt{\#} u_{n-1}
\texttt{\#}u_\infty\texttt{\#}.
\end{displaymath}
When building this solution from left to right,
the partial solution
$ v_1v_{i_2}v_{i_{3}}\ldots v_{i_{j}}$ lags one derivation step behind
the partial  solution
$w_1w_{i_2}w_{i_{3}}\ldots w_{i_{j}}$.
After completing a complete derivation step,
the situation typically looks like this:
\begin{equation} \label{eq:two-rows}
  \begin{array}{@{}l@{}l}
  v_1v_{i_2}v_{i_{3}}\ldots v_{i_{j}}\hfill\null
  &{}=
    \texttt{\#} 
u_0\texttt{\#} u_1 \texttt{\#} \ldots \texttt{\#} u_{s-1}\texttt{\#} 
\\[0.2ex]
w_1w_{i_2}w_{i_{3}}\ldots w_{i_{j}}
  &{}=
    \texttt{\#} 
u_0\texttt{\#} u_1 \texttt{\#} \ldots \texttt{\#} u_{s-1}  \texttt{\#} u_{s} \texttt{\#} 
  \end{array}
\end{equation}
We will often use the image of building a solution incrementally from
left to right in two rows as in~\eqref{eq:two-rows}. The words $v_i$ are appended one after the
other in the upper row as the corresponding words $w_i$ are simultaneously
concatenated in the lower row.

The start pair $(v_1,w_1)$ establishes the situation
\eqref{eq:two-rows}
at the beginning.
When extending this solution, the next word pair
$(v_{i_{j+1}},w_{i_{j+1}})$ and the following ones must match
their $v$-components with the string $u_s
$ in the lower row, and they can do so either by
simply using copying pairs, or by using the left side $l$ of a
rule pair $(l,r)$, in which case $r$ is substituted for $l$ in
the growing word  $u_{s+1}$ in the lower row.
The solution can only be completed with the 
end pair $(  \texttt{\#}u_\infty\texttt{\#},
   \texttt{\#})$, i.e., if $u_s=u_\infty$.

 The following statements can be easily established by induction:  
 \begin{claim}
   \label{prop:PCP-properties}
     \begin{enumerate}
     \item At any time after the start 
       pair
 $(v_1,w_1)$ is included and before
       the last pair $(v_2,w_2)$ is used,
       the partial solution
       $w_1w_{i_2}w_{i_{3}}\ldots w_{i_{j}}$
       in the lower row
contains one more symbol \texttt{\textup\#} 
than
$v_1v_{i_2}v_{i_{3}}\ldots v_{i_{j}}$
in the upper row.
\item \label{shorter}
In particular,
$v_1v_{i_2}v_{i_{3}}\ldots v_{i_{j}}$
is a prefix of
$w_1w_{i_2}w_{i_{3}}\ldots w_{i_{j}}$
and strictly shorter than $w_1w_{i_2}w_{i_{3}}\ldots w_{i_{j}}$.
\item
  Two consecutive words $u$ and $u'$ that are delimited by
  \texttt{\textup\#} symbols are related in the following way:
  Any number of disjoint occurrences of the left hand side
  $l$ of some rule $l\longrightarrow r$ in $u$ are replaced by the
  corresponding right-hand side $r$.
\claimqedhere    
     \end{enumerate}
   \end{claim}
   \smallskip
It is thus possible to carry out several disjoint replacements $l\longrightarrow
r$, possibly with different rules,
when creating
$u'$ 
from~$u
$, but these would be permitted in the semi-Thue system as a
sequence of consecutive replacements.
It follows that the GPCP has a solution if and only if
   $u_0 \xrightarrow
   {*} u_\infty$. 

We see from this discussion that the successive
   words $u_0,u_1,\ldots$ in the solution word \eqref{eq:two-rows}
   are not necessarily the same as the successive words
   $u_0,u_1,\ldots$ in the derivation~\eqref{eq:derivation}:
   They form a subsequence of~\eqref{eq:derivation}, possibly with repetitions.
We note in particular that uniqueness is lost in the transformation.
Even if the semi-Thue system 
has a unique derivation~\eqref{eq:derivation},
the GPCP will
not have unique solutions. The reason is that the GPCP may decide,
at any time, to
use only copying pairs in one round, effectively setting $u_{s+1}=u_s$ in~\eqref{eq:two-rows}.
 \paragraph{Step 2. Encoding the strings in binary.}
\label{para:step2}

 We choose a binary encoding
 $\Sigma' 
 \to \{\mathtt{a},\mathtt{b}\}$ 
  of the alphabet,
 using the codewords
 $\mathtt{bab},\mathtt{baab},\mathtt{baaab},\mathtt{baaaab},\ldots$.
 We denote the encoded version 
 of a string~$w$ 
 by $\langle w\rangle \in \{\mathtt{a},\mathtt{b}\}^*$.

 The advantage of the binary alphabet is that we need just two copying pairs.
All other word pairs are simply encoded symbol by symbol.
Thus, we have the following $k+4$ word pairs:
\begin{itemize}
\item the \emph{start pair} $(v_1,w_1)= ( \langle\texttt{\#}
\rangle, \langle\texttt{\#}  u_0\texttt{\#}\rangle)$
\item the \emph{end pair} $
(v_2,w_2)= 
  ( \langle \texttt{\#}u_\infty\texttt{\#}\rangle,
   \langle\texttt{\#}\rangle)$  
  \item
    $k$    \emph{rule pairs} $(\langle l\rangle,\langle r\rangle)$, one
for each rule $l\longrightarrow r$
\item
 two \emph{copying pairs} $(\mathtt{a},\mathtt{a})$
and $(\mathtt{b},\mathtt{b})$.
\end{itemize}

To ensure that this transformation is correct, one needs
to use the fact that the family of codewords
$\mathtt{bab},\mathtt{baab},\mathtt{baaab},\ldots$
has some strong uniqueness properties.
Not only is the code injective, but
 no codeword $\langle \alpha \rangle$ is a substring of
another codeword or even of a composition of codewords, unless
it appears there because it was used to encode
the symbol $\alpha$. Here is a more formal statement of this property:

\begin{proposition}
  If $\langle u \rangle = x\langle \alpha \rangle y$ for
  some symbol $\alpha\in
  \Sigma' 
  $
  and $x,y\in \{\mathtt{a},\mathtt{b}\}^*$,
  then there are strings
$v,w \in (
\Sigma' 
)^*$ such that
   $\langle v\rangle = x$ and
   $\langle w\rangle = y$ \textup(and consequently,
    $u=v\alpha w$\textup).
\qed
\end{proposition}
I don't know if this property of an encoding has a name.%
\footnote
{Claus 
\cite[p.~57, transformation from system $U_1$ to $U_2$]{\clausEATCS}
applies a code 
  with three codewords
  $h(x_1)=
  \texttt{01}$,
 $h(x_2)=
 \texttt{011}$,
 $h(x_3)=
 \texttt{0111}$ at this place of the argument
 (assuming that the initial semi-Thue 
 system uses a binary alphabet,
 and by adding \texttt{\#}, the alphabet has grown to three letters $x_1,x_2,x_3$).
 This code is injective, but it is inadequate and can lead to errors.
For example, the first three bits of
 \texttt{0111} can be mistaken as the codeword
\texttt{011} for $x_2$, and a rule for a word ending in $x_2$
can be applied. The remaing bit
\texttt{1} is simply copied by the copying pair.
The rule might generate a word ending in $x_1$, and
  $h(x_1)=
  \texttt{01}$
  together with the remaining~\texttt{1} forms
the codeword  \texttt{011}, which stands for~$x_2$. So the derivation
could proceed and the mistake
would go undetected.
%
}

\paragraph{Step 3: Making the rule pairs nonempty.}
\label{para:step3}
When one of the sides $l$ or $r$ of a rule
$l\longrightarrow r$ is empty, the corresponding word in the word pair
 $(\langle l\rangle,\langle r\rangle)$
is also
 empty.
 To avoid empty words, we change the rule pairs to
 $(\langle l\rangle\mathtt{b},\langle r\rangle\mathtt{b})$
for each rule $l\longrightarrow r$.


 The extra letter \texttt{b} makes no difference, because
all codewords start with~\texttt{b}.
Thus, whenever we want to match
$\langle l\rangle
$ in the upper row, we are sure
that the next symbol is~$\mathtt{b}$, so
we might as well copy the $\mathtt{b}$ from the upper row and insert
it after
$\langle r\rangle
$ on the lower row.
We might lose an opportunity to carry out several rule replacements
simultaneously, but this does not change the solvability of the GPCP.
The only problematic case would be that
$\langle l\rangle$ or $\langle r\rangle$ occurs at the end of the
solution string,
but this is impossible since the rules of the GPCP require the string to end with the last words
$v_2=\langle \texttt{\#}u_\infty\texttt{\#}\rangle$ and $w_2=  \langle\texttt{\#}\rangle$.


At this point, we have proved the theorem as far as the GPCP is
concerned.
We only need to observe that $w_1$ is the only word that depends on
$u_0$, the input to the individual word problem for the semi-Thue system. All other
words are fixed.



\paragraph{Optional Step 4: Forcing the start and end pairs to their
  places.}
\label{para:step4}

%
%
For converting the GPCP to an
ordinary PCP,
we use the following construction of \cite[Theorem 3.2]{Harju1997}.
%
We add three new symbols \texttt{\middleseparator },
 \texttt{\leftbound},
and \texttt{\rightbound}
to the alphabet and define two morphisms $\sigma,\rho\colon
\{\texttt{a},\texttt{b}\}^*
\to \{\texttt{a},\texttt{b}
,\texttt{\middleseparator }\}^*
$.
The morphism
$\sigma(u)$ adds a~\texttt{\middleseparator } on the left of every symbol of~$u$,
and $\rho(u)$ adds a \texttt{\middleseparator } on the right of every symbol of~$u$.
For example,
$\sigma(\texttt{aab})=\texttt{\middleseparator a\middleseparator a\middleseparator b}$
and $\rho(\texttt{aab})=\texttt{a\middleseparator a\middleseparator b\middleseparator }$.



We define the modified word pairs $(\bar v_i,\bar w_i)$ as follows:
\begin{align*}
  (\bar v_1,\bar w_1)
  &=(\texttt{\leftbound}\sigma(v_1),\texttt{\leftbound\middleseparator }\rho(w_1))
 =(\texttt{\leftbound}\sigma(v_1),\texttt{\leftbound}\sigma(w_1)\texttt{\middleseparator })
\\  (\bar v_2,\bar w_2)
  &=(\sigma(v_2)\texttt{\middleseparator }\texttt{\rightbound},
 \rho(w_2)\texttt{\rightbound})
=(\texttt{\middleseparator }\rho(v_2)\texttt{\rightbound},
\rho(w_2)\texttt{\rightbound})
\\  (\bar v_i,\bar w_i)
  &=(\sigma(v_i),\rho(w_i)) \quad (i\ge 3)
\end{align*}
 The pattern of these words is shown in
\autoref{tab:word-patterns}.

\begin{table}[h]
{  \centering
  \begin{tabular}{|l|l|c|r|}
    \hline
    &$i=1$& $3\le i\le k+4$& $i=2$\\
    \hline
    $\bar v_i$ & $\texttt{\leftbound\middleseparator }x\texttt{\middleseparator }x\texttt{\middleseparator }x\texttt{\middleseparator }x
                 $ 
& $\texttt{\middleseparator }x\texttt{\middleseparator }x\texttt{\middleseparator }x
                 $ 
& $\texttt{\middleseparator }x\texttt{\middleseparator }x\texttt{\middleseparator }\texttt{\rightbound}
                 $ 
                 \\
    $\bar w_i$ & $\texttt{\leftbound\middleseparator }x\texttt{\middleseparator }x\texttt{\middleseparator }
                 $ 
& $x\texttt{\middleseparator }x\texttt{\middleseparator }x\texttt{\middleseparator }x\texttt{\middleseparator }x\texttt{\middleseparator }
                     $ 
                     &$x\texttt{\middleseparator }\texttt{\rightbound}$
                 \\\hline
  \end{tabular}

  }
  \caption{Patterns of
    some representative 
    words $\bar v_i$ and $\bar w_i$.
The symbol $x$ stands for \texttt{a} or \texttt{b}.   
    The lengths of these words can of course be different from what is
    shown:
    The number of $x$'s is the length of the original words
    $v_i$ or~$w_i$.
}
  \label{tab:word-patterns}
\end{table}


It is clear that every solution to the GPCP gives rise
to a solution of the PCP:
For every initial part
 $1,i_2,\ldots, i_{j}$ of a PCP solution that does not include the
 final word pair $i_m=2$, the following relation can be easily
 proved by induction.
 \begin{align}
   \label{eq:relation-upper-row}
             \bar v_1\bar v_{i_2}\bar v_{i_{3}}\ldots \bar v_{i_{j}}
  &=\texttt{\leftbound}\,
    \sigma (  v_1 v_{i_2} v_{i_{3}}\ldots  v_{i_{j}})
           \\
           \bar w_1\bar w_{i_2}\bar w_{i_{3}}\ldots \bar w_{i_{j}}
  &{}= \texttt{\leftbound}\,
    \sigma( w_1 w_{i_2} w_{i_{3}}\ldots  w_{i_{j}})
    \,\texttt{\textup\middleseparator }
\label{eq:relation-lower-row}
 \end{align}
 When the solution includes the  final pair $i_m=2$,
 the extra \texttt{\middleseparator } 
 appears also at the end of
 \eqref
 {eq:relation-upper-row}, and
the closing bracket sign \texttt{\rightbound} is added to both strings,
so we have equality between
 \eqref
{eq:relation-upper-row} and~\eqref
{eq:relation-lower-row}.

For the reverse direction,
we have to argue that a solution to the PCP for the word pairs
 $(\bar v_i,\bar w_i)$
gives rise to the a solution of the GPCP.
Since in every modified pair  $(\bar v_i,\bar w_i)$ except
$(\bar v_1,\bar w_1)$, the two words
$\bar v_i$ and $\bar w_j$ start with different letters,
the start pair $(\bar v_1,\bar w_1)$ is the only pair that can
be used at the start of the PCP.
(This argument relies on the fact that all words are
nonempty.)

Moreover, the opening bracket \texttt{\leftbound} 
occurs only in 
$\bar v_1$
and $\bar w_1$.
Therefore, if the pair
$(\bar v_1,\bar w_1)$ is used a second time,
these occurrences of $\bar v_1$
and
 $\bar w_1$
must be matched
in a left-aligned position.
 We can then cut the solution
 $i_1i_2\ldots i_m$
 at this point,
 and both pieces will be
 solutions of the PCP.\footnote{
 In the construction of
\cite[Lemma~1]{halava07-Claus-instances},
the initial brackets \texttt{\leftbound} 
in
$\bar v_1$
and $\bar w_1$ are omitted
and the final brackets \texttt{\rightbound} 
in 
$\bar v_2$
and $\bar w_2$
are replaced by \texttt{\middleseparator } symbols,
which are denoted by $d$ in
\cite{halava07-Claus-instances}.
 The statement is then still true, but only
 for nonempty words $v_i$ and $w_i$,
see \autoref{sec:counter-halava}.
 This version economizes in alphabet size, but correctness is 
 more delicate.
}
Thus, by cutting the solution before the first repeated use of
the pair $(\bar v_1,\bar w_1)$ (if any),
we can assume that we only consider a solution
where 
 $(\bar v_1,\bar w_1)$ is used at the beginning and nowhere else.

 Looking at the last symbol of the solution,
 the symmetric argument
 establishes that $(\bar v_2,\bar w_2)$ is
used at the end and nowhere else.



 We have seen that every solution of
 the PCP with the word pairs
 $(\bar v_i,\bar w_i)$, possibly after cutting it before the second
 use
 of the pair $(\bar v_1,\bar w_1)$, must satisfy the constraints of
 the GPCP,
 and by
\thetag{\ref
{eq:relation-upper-row}--\ref
{eq:relation-lower-row}},
this gives a solution for the original GPCP.
(This statement holds also if the pairs may contain empty words, with
the obvious exception of the pair $(\epsilon,
\epsilon)$, but this requires additional arguments.)

   \paragraph{Step 5: Reduction to two characters.}
\label{para:step5}

   Step~4 has added three characters $\texttt{\leftbound},
 \texttt{\rightbound},\texttt{\middleseparator }$
   to the alphabet.
   We can easily reduce $\{\texttt{a},\texttt{b}
   , \texttt{\leftbound},\texttt{\rightbound},\texttt{\middleseparator }\}$
   to a binary alphabet, for example
   with the fixed-length codewords
\texttt{aaa},
\texttt{bbb},
\texttt{bba},
\texttt{aba}, and 
\texttt{bab}. 
This concludes the proof of
\autoref{lem:claus-instance}.
\end{proof}


\subsection{Counterexample to
    Lemma 1 in
    Halava, Harju, and Hirvensalo~%
    \texorpdfstring{\cite
      {halava07-Claus-instances}}{  (2007)}}
\label{sec:counter-halava}

Our proof approach differs slightly from the proof of Claus
\cite{\clausEATCS}:
There, the
transformation to the PCP (\hyperref[para:step1]{Step~1})
is merged
with the enforcement of the
GPCP rules
 (\refToStepA4)
 into one transformation.
The binary coding
 (\refToStepA2)
 is done \emph{before} the transformation to the PCP
 (\refToStepA1),
already at the level of
the semi-Thue system.

  Halava, Harju, and Hirvensalo 
\cite[Theorem~4]
{halava07-Claus-instances}
gave an independent presentation of Claus' construction.\footnote
{See also the
TUCS Technical Report
No.~766, Turku Centre for Computer Science, April 2006, by the same authors with the same title,
\url{https://typeset.io/pdf/undecidability-bounds-for-integer-matrices-using-claus-1nwnni898j.pdf}}
Unfortunately, their proof
mingles the different construction steps (Steps \refToStep1, \refToStep2, \refToStep4, and~\refToStep5) into a
single condensed one-shot construction,
which thereby becomes less
transparent.
For the correctness, they refer to on a lemma about PCP solutions, which, however,
is subject to exceptions when empty words appear.
(No proof of that lemma is given in \cite{halava07-Claus-instances}.
The authors write that the
lemma ``is straightforward'' and
supply an
unspecific reference to a 70-pages handbook article
\cite{Harju1997} of Harju and Karhumäki.)\footnote
{The handbook article of
 Harju and Karhumäki formulates Step 4 as a separate theorem
 \cite[Theorem~3.2]{Harju1997}.
Still, 
they prove their version of Theorem~\ref{lem:claus-instance}
\cite[Theorem 4.1]{Harju1997} by a ``refinement of the general ideas
presented in the
proof of Theorem~3.2'', instead of simply \emph{applying} that theorem.}

It is easy to construct a counterexample when one of the words is empty.
In the setting of \cite{halava07-Claus-instances},
a PCP with $n$ word pairs $(v_i,w_i)$ is specified by two
morphisms $g,h\colon
\{b_1,b_2,\ldots,b_n\}\to \Gamma^*
$ such that $(h(b_i),g(b_i))=(v_i,w_i)$.
In line with \cite{halava07-Claus-instances}, we translate
the PCP into a binary alphabet with codewords of
the form $\mathtt{ab}^{i+1}\mathtt{a}$, although
this translation is not important for our point.
We use the following code
$\varphi\colon \Gamma^* \to\{\mathtt{a},\mathtt{b}\}^*$:
$\varphi
(\texttt{\middleseparator })=\texttt{aba}=:{d}$,
$\varphi
(\texttt{u})=\texttt{abba}$,
$\varphi
(\texttt{v})=\texttt{abbba}$,
$\varphi
(\texttt{w})=\texttt{abbbba}$.

The PCP of our example has three word pairs. So
we use a 3-letter index alphabet
$\{b_1,b_2,b_3=b_n\}$, and the word pairs are:
\begin{align*}
  g(b_1)&=\varphi(\texttt{\middleseparator })={d} &  h(b_1)&=\varphi(\texttt{\middleseparator v\middleseparator w})
  &&\text{(the start pair)}\hskip 4cm\\
  g(b_2)&=\varphi(\epsilon)=\epsilon &  h(b_2)&=\varphi(\texttt{\middleseparator u})\\
  g(b_3)&=\varphi(\texttt{u\middleseparator v\middleseparator w\middleseparator \middleseparator }) &  h(b_3)&=\varphi(\texttt{\middleseparator \middleseparator })={dd}
  &&\text{(the end pair)}
\end{align*}
It can be checked that they
follow the pattern that is prescribed for
what is called a \emph{Claus instance}
in~\cite{halava07-Claus-instances}.
Lemma~1 in \cite{halava07-Claus-instances} claims in particular that, for such an
instance, all nonempty solutions must start with $b_1$. 

However, $g(b_2b_1b_3)
=h(b_2b_1b_3)=
\varphi(\texttt{\middleseparator u\middleseparator v\middleseparator w\middleseparator \middleseparator })
$.
So,
$
b_2b_1b_3$ is a solution that does not start with $b_1$.
The reason for the failure is that, since $g(b_2)$ is empty, the argument that the solution
cannot start with anything but $b_1$ is not valid.

The original proof
of Claus~\cite{\clausEATCS}
is correct in this point, since it is more
generous with additional alphabet symbols. (Our
extra symbol~\texttt{*} is called $\beta$ in the notation
of~\cite{\clausEATCS};
the role of the bracket symbols \texttt{\leftbound}
and \texttt{\rightbound} is played by $h(\gamma)=\mathtt{0111}$%
.)

To be fair, one should mention that
\cite{halava07-Claus-instances} sketches an alternative correctness
proof of the construction for their Theorem~4
(our Theorem~\ref{lem:claus-instance}), which is tailored to the PCP at hand
and which is independent of Lemma~1. That proof, however, also
glosses over 
the case of empty words.






\ifnotkurz

\medbreak
 
\fi

\clearpage

\section{Program for checking the multiplicative law (\autoref{lem:multi}) in \texttt{sagemath}}
\label{sec:check}

\nolinenumbers

\begin{verbatim}
from sage.symbolic.function_factory import function

BASE = 10

sigma = function("sigma",nargs=1) # unevaluated function
def sigma_num(u): return float(int("0"+u,BASE)/BASE**len(u))
    # parameter u is a string; returns 0.u as a number
    
l=function("l",nargs=1) # unevaluated function
def l_num(u): return float(BASE**-len(u))  # 10^{-|u|}

var("v w")
A0 = matrix([ (1, 0, 0, 0, 0, 0),
    (sigma(v), l(v), 0, 0, 0, 0),
    (sigma(v)^2, 2*l(v)*sigma(v), l(v)^2, 0, 0, 0),
    (sigma(w), 0, 0, l(w), 0, 0),
    (sigma(w)^2, 0, 0, 2*l(w)*sigma(w), l(w)^2, 0),
    (sigma(v)*sigma(w), l(v)*sigma(w), 0, l(w)*sigma(v), 0, l(v)*l(w)) ])

print("A0(v,w)="); print(A0,"\n")

def A0_numeric(vv,ww): return A0.subs(
        [l(v) == l_num(vv), sigma(v) == sigma_num(vv),
         l(w) == l_num(ww), sigma(w) == sigma_num(ww)] )
    # substitution of numerical values
     
v1,w1="12","999"
v2,w2="654","77"
print("A0("+v1+","+w1+")="); print(A0_numeric(v1,w1),"\n")

print("test multiplicative law, up to machine precision:")
diff = A0_numeric(v1,w1)*A0_numeric(v2,w2)-A0_numeric(v1+v2,w1+w2)
print("max deviation:", max(map(abs,diff.list())))
\end{verbatim}


This script 
can be directly fed as input to \texttt{sage}\footnote
{\url{https://www.sagemath.org/}}. It is
robust
against the loss of proper indentation.

\end{document}